\documentclass[11pt]{article}
\usepackage{graphicx}
\usepackage[margin=1in]{geometry}
\usepackage{amsmath}
\usepackage{amssymb}
\usepackage{amsthm}
\usepackage{accents}
\usepackage[maxfloats=128]{morefloats}
\usepackage[breaklinks,hidelinks]{hyperref}
\newcommand{\sgn}{\mathop{\mathrm{sgn}}}

\newcommand{\E}{{\bf E}}

\newtheorem{theorem}{Theorem}[]
\newtheorem{lemma}[theorem]{Lemma}
\newtheorem{corollary}[theorem]{Corollary}

\newtheorem{remark1}[theorem]{Remark}

\renewcommand{\epsilon}{\varepsilon}

\title{Secure multiparty computations in floating-point arithmetic}
\author{Chuan Guo, Awni Hannun, Brian Knott, Laurens van der Maaten,\\
Mark Tygert, and Ruiyu Zhu}

\begin{document}

\maketitle

\begin{abstract}
Secure multiparty computations enable the distribution
of so-called shares of sensitive data to multiple parties
such that the multiple parties can effectively process the data
while being unable to glean much information about the data
(at least not without collusion among all parties
to put back together all the shares).
Thus, the parties may conspire to send all their processed results
to a trusted third party (perhaps the data provider)
at the conclusion of the computations,
with only the trusted third party being able to view the final results.
Secure multiparty computations for privacy-preserving machine-learning
turn out to be possible using solely standard floating-point arithmetic,
at least with a carefully controlled leakage of information less than
the loss of accuracy due to roundoff, all backed
by rigorous mathematical proofs of worst-case bounds on information loss
and numerical stability in finite-precision arithmetic.
Numerical examples illustrate the high performance
attained on commodity off-the-shelf hardware for generalized linear models,
including ordinary linear least-squares regression,
binary and multinomial logistic regression, probit regression,
and Poisson regression.
\end{abstract}

\section{Introduction}

Passwords and long account and credit-card numbers
are the dominant security measures,
not because they are the most secure,
but because they are the most conveniently implemented.
Some data demands the highest levels of security and privacy protections,
while for other data processing efficiency and sheer convenience are paramount
--- some security is better than none (which tends to be the alternative).
The present paper proposes privacy-preserving, secure multiparty computations
performed solely in the IEEE standard double-precision arithmetic that
dominates most platforms for numerical computations.
The scheme amounts to lossy, leaky cryptography,
with the loss of accuracy and leakage of information
carefully controlled via mathematical analysis and rigorous proofs.
Information loss balances against roundoff error,
providing perfect privacy at a specified finite precision of computations.

Perfect privacy at a given precision is when the information leakage
is less than the specified precision
(precision being limited due to roundoff error).
The present paper provides perfect privacy at a precision of about $10^{-5}$
in the IEEE standard double-precision arithmetic of~\cite{ieee};
observing the encrypted outputs
leaks no more than a millionth of a bit per input real number,
whereas roundoff alters the results by around one part in a hundred thousand.
In a megapixel image, the encrypted image would leak at most a single bit ---
enough information to discern whether the original image is dim or bright,
perhaps, but no more. Performing all computations in floating-point arithmetic
facilitates implementations on existing hardware, including both commodity
central-processing units (CPUs) and graphics-processing units (GPUs),
whereas alternative methods based on integer modular arithmetic could require
difficult specialized optimizations to attain performance on par
with the scheme proposed in this paper (even then, schemes based
on integer modular arithmetic would have to contend with tricky issues
of discretization and precision in order to handle the real numbers
required for machine learning and statistics).

The algorithms and analysis consider the traditional {\it honest-but-curious}
model of threats:
we assume that the multiple parties follow the agreed-upon protocols correctly
but may try to glean information from data they observe;
the secure multiparty computations prevent any of the parties from gleaning
much information without all parties conspiring together to break the scheme.

Our analysis provides no guarantees about what information leaks
when all parties collude to reveal encrypted results.
If all parties conspire to collect together all their shares
or send them to a collecting agency, then the unified collection
will reveal the secrets of whatever results get collected.
If the collected information results from training a machine-learned model,
then revealing the trained model can compromise the confidentiality
of the data used to train that model,
unless the model is differentially private.
Ensuring privacy even after revealing the results
of secure multiparty computations is complementary
to securing the intermediate computations.
The present paper only guarantees the privacy of the intermediate computations,
providing no guarantees about what leaks when all parties collude
to reveal the final results of their secure multiparty computations.

This paper has the following structure:
Section~\ref{smpc} introduces secure multiparty computations
in floating-point arithmetic, reviewing classical methods
such as additive sharing and Beaver multiplication.
Section~\ref{leakage} upper-bounds the amount of information that can leak,
referring to Appendices~\ref{mainproof}, \ref{simpleproof}, and~\ref{detailed}
for full, rigorous proofs.
Section~\ref{polynomials} reviews techniques for efficient, highly accurate
polynomial approximations to many real functions of interest
(notably those in Table~\ref{polytable}).
Section~\ref{numex} validates an implementation on synthetic examples
and illustrates its performance on real measured data, too;
the examples apply various generalized linear models,
including ordinary linear least-squares regression,
binary and multinomial logistic regression, probit regression,
and Poisson regression.
Appendices~\ref{chebappendix}, \ref{sgdrev}, and~\ref{glmrev}
very briefly review Chebyshev series, minibatched stochastic gradient descent,
and generalized linear models, respectively --- readers may wish to refer
to those appendices as concise refreshers.

Throughout, all numbers and random variables are real-valued,
even when not stated explicitly.

\section{Secure multiparty computations}
\label{smpc}

Secure multiparty computations allow holders of sensitive data
to securely distribute so-called shares of their data to multiple parties
such that the multiple parties can process the data without revealing the data
and can only reconstruct the data by colluding to put back together
(that is, to sum) all the shares. We briefly review an arithmetic scheme
in the present section.
The arithmetic scheme supports addition and multiplication,
as discussed in the present section, as well as functions that polynomials
can approximate accurately, as discussed in Section~\ref{polynomials} below.
To be concrete and simplify the presentation,
we focus first on secure two-party computations,
in Subsection~\ref{two-party}, then sketch an extension to several parties
in Subsection~\ref{several}.

\subsection{Two-party computations}
\label{two-party}

We can hide a matrix $X$ by masking with a random matrix $Y$,
so that one party holds $X-Y$ and the other party holds $Y$.
Given other data, say $U$, and another random matrix $V$ hiding $U$,
so that one party holds $U-V$ and the other party holds $V$,
the parties can then independently form the sums
$(U-V)+(X-Y)$ and $V+Y$ required to reconstruct
$U+X = (U-V)+(X-Y)+(V+Y)$ if all these matrices have the same dimensions.
This is known as ``additive sharing,''
as described, for example, by~\cite{bogdanov-laur-willemson}.
Additive sharing thus supports privacy-preserving addition of $U$ and $X$.

As introduced by~\cite{beaver}, used by~\cite{bogdanov-laur-willemson},
and reviewed in Table~\ref{Beaver},
privacy-preserving multiplication of $U$ and $X$
is also possible whenever $U$ and $X$ are matrices such that their product $UX$
is well-defined.
Summing across the two parties in line~10 of Table~\ref{Beaver}
would yield $PR + (U-P)R + P(X-R) + (U-P)(X-R) = UX$, as desired.
Notice that $U-P$ and $X-R$ in lines~8 and~9 of Table~\ref{Beaver}
are still masked by the random matrices $P$ and $R$ whose values are unknown
to the two parties. The all-reduce in lines~8 and~9 requires
the parties to conspire and distribute to each other the results
of summing their respective shares of lines~6 and~7,
but does not require the parties to get shares of any secret distributed data
--- that was necessary only for the original data $U$ and $X$ being multiplied
(and this ``original'' data can be the result
of prior secure multiparty computations).
Also, all values on lines~3--5 are independent of $U$ and $X$,
so can be precomputed and stored on disk.
As with addition, multiplication via the Beaver scheme never requires
securely distributing secret shares, so long as the input data 
and so-called Beaver triples $(P, R, PR)$ from Table~\ref{Beaver} have already
been securely distributed into shares.
The secure distribution of shares is completely separate
from the processing of the resulting distributed data set.
Communication among the parties during processing is solely
via the all-reduce in lines~8 and~9.

Table~\ref{Beaver}'s scheme also works with matrix multiplication
replaced by convolution of sequences.

\begin{table}
\begin{center}
\begin{tabular}{llll}
line & source & party 1 & party 2 \\\hline
1 & distributed data & $U-V$ & $V$ \\
2 & distributed data & $X-Y$ & $Y$ \\
3 & read from disk & $P-Q$ & $Q$ \\
4 & read from disk & $R-S$ & $S$ \\
5 & read from disk & $PR-T$ & $T$ \\
6 & line 1 $-$ line 3 & $(U-V)-(P-Q)$ & $V-Q$ \\
7 & line 2 $-$ line 4 & $(X-Y)-(R-S)$ & $Y-S$ \\
8 & all reduce line 6 & $U-P$ & $U-P$ \\
9 & all reduce line 7 & $X-R$ & $X-R$ \\
10 & line 5 $+$ (line 8)(line 4) $+$ &
$(PR-T) + (U-P)(R-S)$ $+$ &
$T + (U-P)S$ $+$ \\
& (line 3)(line 9) $+$ & $(P-Q)(X-R)$ $+$ &
$Q(X-R)$ $+$ \\
& (line 8)(line 9)/2 & $(U-P)(X-R)/2$ & $(U-P)(X-R)/2$
\end{tabular}
\end{center}
\caption{Ledgers for two parties in a Beaver multiplication of $U$ and $X$}
\label{Beaver}
\end{table}

Table~\ref{Beaver} simplifies to Table~\ref{squaring} for the case
when $U$ and $X$ are scalars such that $U = X$.
Summing across the two parties in line~6 of Table~\ref{squaring}
yields $P^2 + 2P(X-P) + (X-P)^2 = X^2$, as desired.
Notice that this recovers $X^2$ from a sum involving several terms
as large as $P^2$; if $|X| \le 1 < 3 < \gamma$ and $|P| \le \gamma$
(as well as $|T| \le \gamma^2$, where $T$ cancels when summing across
the two parties in line~6),
then we obtain $X^2$ to precision upper-bounded by $6\gamma^2 \cdot \epsilon$,
where $\epsilon$ denotes the machine precision
($\epsilon$ is approximately $2.2 \times 10^{-16}$
in the IEEE standard double-precision arithmetic of~\cite{ieee}).
Theorem~\ref{Beaverdetails} proves that the information leakage may be
as large as $6/\gamma$ if $P$, $Q$, and $Y$ are distributed uniformly
over $[-\gamma, \gamma]$ and $T$ is distributed uniformly
over $[-\gamma^2, \gamma^2]$
(similarly, $P$, $Q$, $R$, $S$, $V$, and $Y$ in Table~\ref{Beaver}
should be distributed uniformly over $[-\gamma, \gamma]$
while $T$ should be distributed uniformly over $[-\gamma^2, \gamma^2]$).
Balancing the roundoff bound with the information bound requires
$6\gamma^2 \cdot \epsilon = 6/\gamma$, so that $\gamma = 1/\sqrt[3]{\epsilon}$
(so $\gamma \approx 10^5$ for IEEE standard double-precision arithmetic).

\begin{table}
\begin{center}
\begin{tabular}{llll}
line & source & party 1 & party 2 \\\hline
1 & distributed data & $X-Y$ & $Y$ \\
2 & read from disk & $P-Q$ & $Q$ \\
3 & read from disk & $P^2-T$ & $T$ \\
4 & line 1 $-$ line 2 & $(X-Y)-(P-Q)$ & $Y-Q$ \\
5 & all reduce line 4 & $X-P$ & $X-P$ \\
6 & line 3 $+$ (line 2)(line 5) $\cdot$ 2 $+$ &
$(P^2-T) + (P-Q)(X-P) \cdot 2$ $+$ &
$T + Q(X-P) \cdot 2$ $+$ \\
& (line 5)$^2$/2 & $(X-P)^2/2$ & $(X-P)^2/2$
\end{tabular}
\end{center}
\caption{Ledgers for two parties in a Beaver squaring of $X$}
\label{squaring}
\end{table}

\subsection{Several-party computations}
\label{several}

Extending Subsection~\ref{two-party} beyond two parties is straightforward.
The steps in the algorithms, summarized in the columns labeled ``source''
in Tables~\ref{Beaver} and~\ref{squaring}, stay as they were
(the division by 2 in the last line of Table~\ref{Beaver}
and in the last line of Table~\ref{squaring} becomes division
by the number of parties).
Distributing additive shares of data across several parties works as follows:
for each piece of data to be distributed (in machine learning,
a ``piece'' may naturally be a sample or example from the collection
of all samples or examples), we generate $n$ independent
and identically distributed random matrices $Y_1$, $Y_2$, \dots, $Y_n$,
where $n$ is the number of parties (the number of parties need not relate
to the total number of pieces, samples, or examples of data being distributed).
We randomly permute the parties and then distribute to them
(in that random order)
$X + Y_1 - Y_2$, $Y_2 - Y_3$, $Y_3 - Y_4$, \dots, $Y_{n-1} - Y_n$, $Y_n - Y_1$,
where $X$ is the piece of data being shared.
We generate different independent random variables and random permutations
for different pieces of data.
The distribution of the difference between independent random matrices
drawn from the same distribution is the same for each party,
making this an especially simple generalization to the case
of several parties.
Distributing additive shares to several parties leaks
somewhat more information than limiting to only two parties; the present paper
focuses on the case of two parties for simplicity.

\section{Information leakage}
\label{leakage}

This section bounds the amount of information-theoretic entropy that can leak
when adding noise for masking, drawing heavily on canonical concepts
from information theory, as detailed, for example, by~\cite{cover-thomas}.

We denote by $X$ the scalar random variable that we want to hide,
and by $Y$ an independent variate that we add to $X$ to effect the hiding.
To simplify the analysis, we assume that the distribution of $Y$ arises
from a probability density function.
Then, revealing $X+Y$ leaks the following number of bits of information
about $X$:
\begin{equation}
\label{mutual}
H(X) - H(X \;|\; X+Y) = I(X;\; X+Y) = H(X+Y) - H(X+Y \;|\; X) = H(X+Y) - H(Y).
\end{equation}
In the left-hand side of~(\ref{mutual}),
$H$ denotes the Shannon entropy measured in bits if the distribution of $X$
is discrete, and the differential entropy measured in bits (rather than nats)
if the distribution of $X$ is continuous.
In the right-hand sides of~(\ref{mutual}),
$H$ denotes the differential entropy measured in bits (not nats);
$I$ denotes the mutual information.
Given a prior on $X$, Bayes' Rule yields the full posterior distribution
for $X$ given $X+Y$; the information gain (or loss or leakage)
defined in~(\ref{mutual}) is a summary statistic characterizing the divergence
of the posterior from the prior.

Recall that mutual information is the fundamental limit on how much information
can be gleaned from observing the outputs of a noisy channel;
in our setting, we purposefully add noise in order to reveal only the results
of communications via a (very) noisy channel, purposefully obscuring with noise
the signal containing data being kept confidential and secure.

The information leakage is at most $\beta/\gamma$ bits
if $|X| \le \beta < \gamma$ and $Y$ is distributed uniformly
over $[-\gamma, \gamma]$, as stated in the following theorem
and proven in Appendix~\ref{mainproof}:
\begin{theorem}
\label{infoleak}
Suppose that $X$ and $Y$ are independent scalar random variables
and $\beta$ and $\gamma$ are positive real numbers
such that $|X| \le \beta < \gamma$
and $Y$ is distributed uniformly over $[-\gamma, \gamma]$.
Then, the information leaked about $X$ from observing $X+Y$ satisfies
\begin{equation}
\label{worst}
I(X; X+Y) \le \frac{\beta}{\gamma},
\end{equation}
where $I$ denotes the mutual information between $X$ and $X+Y$,
measured in bits (not nats); the mutual information satisfies~(\ref{mutual}),
which expresses $I$ as a change in entropy.
The inequality in~(\ref{worst}) is an equality when $X$ is $\beta$ times
a Rademacher variate.
\end{theorem}

The following theorem, proven in Appendix~\ref{simpleproof},
states that the information leakage from hiding data multiple times
is at most the sum of the information leaking from each individual hiding.
\begin{theorem}
\label{chaining}
Suppose that $X$, $Y$, and $Z$ are independent scalar random variables. Then,
\begin{equation}
\label{subadditivity}
I(X;\, X+Y,\, X+Z) \le I(X;\, X+Y) + I(X;\, X+Z),
\end{equation}
where $I$ denotes the mutual information.
\end{theorem}

The procedures of Tables~\ref{Beaver} and~\ref{squaring}
can also leak information, but not much --- consider Table~\ref{squaring}:
Party~2 observes nothing about the input data $X$ other than $X-P$,
and Theorem~\ref{infoleak} bounds how much information that reveals about $X$.
Party~1 observes $X-Y$, $P-Q$, $P^2-T$, $(X-Y)-(P-Q)$, $X-P$,
and $(P^2-T) + 2(P-Q)(X-P) + (X-P)^2/2$.
The following theorem, proven in Appendix~\ref{detailed}, bounds
how much information about $X$ these observations reveal.
\begin{theorem}
\label{Beaverdetails}
Suppose that $X$, $Y$, $P$, $Q$, and $T$
are independent scalar random variables and $\gamma$ is a positive real number
such that $|X| \le 1 < 3 < \gamma$, the random variable $T$
is distributed uniformly over $[-\gamma^2, \gamma^2]$,
and $Y$, $P$, and $Q$ are distributed uniformly over $[-\gamma, \gamma]$. Then,
\begin{multline}
\label{mainsquare}
I\Bigl(X;\,
X-Y,\, P-Q,\, P^2-T,\, (X-Y)-(P-Q),\, X-P,\,
(P^2-T) + 2(P-Q)(X-P) + (X-P)^2/2\Bigr) \\
\le \frac{5}{\gamma} + \frac{1}{\gamma^2},
\end{multline}
where $I$ denotes the mutual information measured in bits,
and its arguments (aside from $X$) are the observations in Table~\ref{squaring}
under the column for ``party~1.''
\end{theorem}

The following theorem states the classical data-processing inequality:
\begin{theorem}
\label{dataproc}
Suppose that random matrices $X$ and $Z$ are conditionally independent
given a random matrix $Y$. Then,
\begin{equation}
I(X; Z) \le I(X; Y)
\end{equation}
and
\begin{equation}
I(X; Z) \le I(Y; Z),
\end{equation}
where $I$ denotes the mutual information.
\end{theorem}

The following is a corollary of Theorem~\ref{dataproc}:
\begin{corollary}
Suppose that $X$ and $Y$ are random matrices
and $f$ is a deterministic function. Then,
\begin{equation}
\label{complicated}
I(X; f(Y)) \le I(X; Y),
\end{equation}
where $I$ denotes the mutual information.
\end{corollary}

Combining~(\ref{mutual}) and~(\ref{complicated})
shows that even very complicated manipulations such as the iterations
in Section~\ref{polynomials} below cause no further information leakage,
despite changing the added noise in some highly nonlinear fashion:
(\ref{complicated}) guarantees that no information leaks
beyond the individual maskings obeying~(\ref{worst})--(\ref{mainsquare}),
so long as the manipulations are deterministic algorithms
(or randomized algorithms with randomization independent
of the data being masked).

\section{Polynomial approximations}
\label{polynomials}

Polynomials can approximate many functions of interest in machine learning,
allowing the accurate approximation of those functions using only additions
and multiplications. Section~\ref{smpc} above discusses schemes that
multiple parties can use to perform additions and multiplications securely.
The present section describes polynomial approximations useful in tandem
with the schemes of Section~\ref{smpc}.
Subsection~\ref{Newton} leverages the method of Newton and Raphson.
Subsection~\ref{Chebyshev} uses Chebyshev series.
Subsection~\ref{softmax} utilizes Pad\'e approximation,
in the method of scaling and squaring.
The method of Newton and Raphson tends to be the most efficient,
while Chebyshev series apply to a much broader class of functions.
The method of scaling and squaring is for exponentiation.
Table~\ref{polytable} lists the functions that each subsection treats.

\begin{table}
\begin{center}
\begin{tabular}{lll}
$f(x)$ & Name & Subsection \\\hline
$\sgn(x)$ & sign or signum & \ref{Newton} \\
$|x| = x \sgn(x)$ & absolute value & \ref{Newton} \\
$1/x$ & reciprocal & \ref{Newton} \\
$1/\sqrt{x}$ & raise to $-1/2$ power & \ref{Newton} \\
$x^{-1/8}$ & raise to $-1/8$ power & \ref{Newton} \\
$\operatorname{ReLU}(x) = \max(x, 0)$ & rectified linear unit & \ref{Newton} \\
$\tanh(x)$ & hyperbolic tangent & \ref{Chebyshev} \\
$1/(1 + \exp(-x))$ & logistic & \ref{Chebyshev} \\
$\int_{-\infty}^x \exp(-y^2/2) \, dy \bigm/ \sqrt{2\pi}$
& CDF of standard normal & \ref{Chebyshev} \\
$\exp(x)$ & exponential & \ref{softmax}
\end{tabular}
\end{center}
\caption{Subsections providing polynomial approximations to various functions}
\label{polytable}
\end{table}

\subsection{Newton iterations}
\label{Newton}

Various iterations derived from the Newton method for finding zeros
of functions allow the computation of functions such as $\sgn(x)$,
$1/x$, and $1/\sqrt{x}$ using only additions and multiplications
(not requiring any divisions or square roots); in this subsection,
$x$ denotes a real number.

According to~\cite{kenney-laub},
the Newton-Schulz iterations for computing $\sgn(x)$ are
\begin{equation}
\label{sgn}
y_{k+1} = y_k (3 - y_k^2) / 2,
\end{equation}
with $y_0 = x / \gamma$, where $|x| \le \gamma$
(and the desired loss of accuracy relative to the machine precision
is less than a factor of $\gamma$).

According to~\cite{kenney-laub},
the Schulz (or Newton) iterations for computing $1/x$ when $x>0$ are
\begin{equation}
\label{inverse}
y_{k+1} = y_k (2 - x y_k),
\end{equation}
with $y_0 = 1$. Rescaling $x$ (and then adjusting the resulting reciprocal)
is important to align with the domain of convergence and high accuracy
illustrated in Figure~\ref{reciprocal}; and similar observations pertain
to the rest of the iterations of the present subsection.
In Subsection~\ref{softmax} below, $x$ can range
from 1 to the number of terms in the softmax (so requires scaling
by the reciprocal of the number of terms in the softmax).

According to~\cite{guo-higham},
the Newton iterations for computing $1/\sqrt{x}$ when $x>0$ are
\begin{equation}
\label{invs}
y_{k+1} = y_k (3 - x y_k^2) / 2,
\end{equation}
with $y_0 = 1$; similarly, the Newton iterations for computing $x^{-1/8}$
when $x>0$ are
\begin{equation}
\label{inv8}
y_{k+1} = y_k (9 - x y_k^8) / 8,
\end{equation}
with $y_0 = 1$.

Figures~\ref{reciprocal}--\ref{absval} illustrate the errors
obtained from~(\ref{sgn})--(\ref{inv8});
note that the scale of the vertical axes
in Figures~\ref{reciprocal}--\ref{inv8root} involve 1e--16.
In the figures, the tilde denotes the approximation computed
via the Newton iterations~(\ref{sgn})--(\ref{inv8});
for example, $\widetilde{1/x}$ approximates $1/x$.

A common operation in the deep learning
of~\cite{lecun-chopra-hadsell-ranzato-huang} and others
is the rectified linear unit
\begin{equation}
\label{relu}
\operatorname{ReLU}(x) = \max(x, 0) = \frac{x (1 + \sgn(x))}{2},
\end{equation}
easily obtained from~(\ref{sgn}).

\begin{figure}
\begin{centering}

\parbox{.81\textwidth}{\includegraphics[width=.8\textwidth]{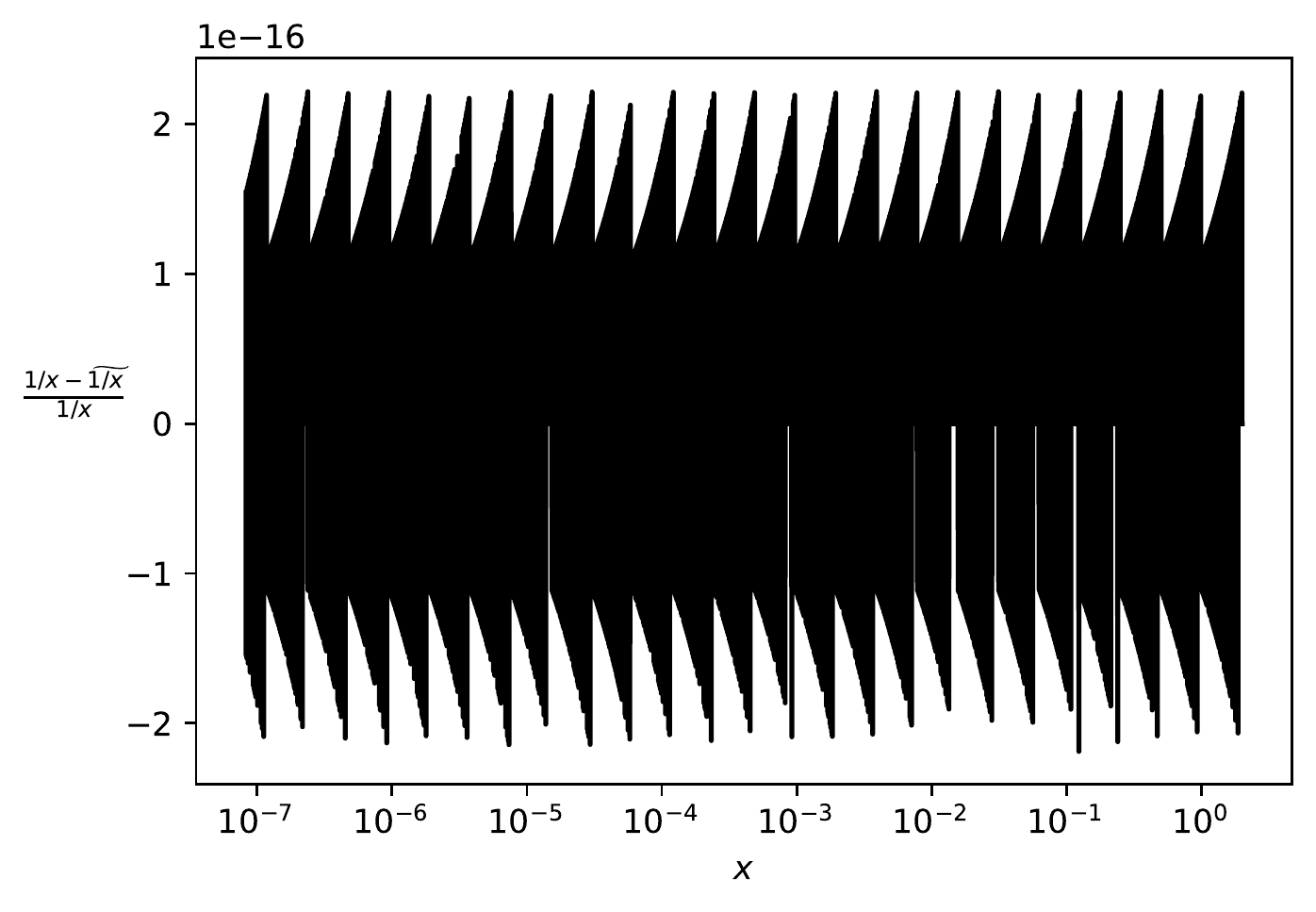}}

\end{centering}
\caption{Relative error in computation of $1/x$ with 30 iterations
of~(\ref{inverse})}
\label{reciprocal}
\end{figure}

\begin{figure}
\begin{centering}

\parbox{.81\textwidth}{\includegraphics[width=.8\textwidth]{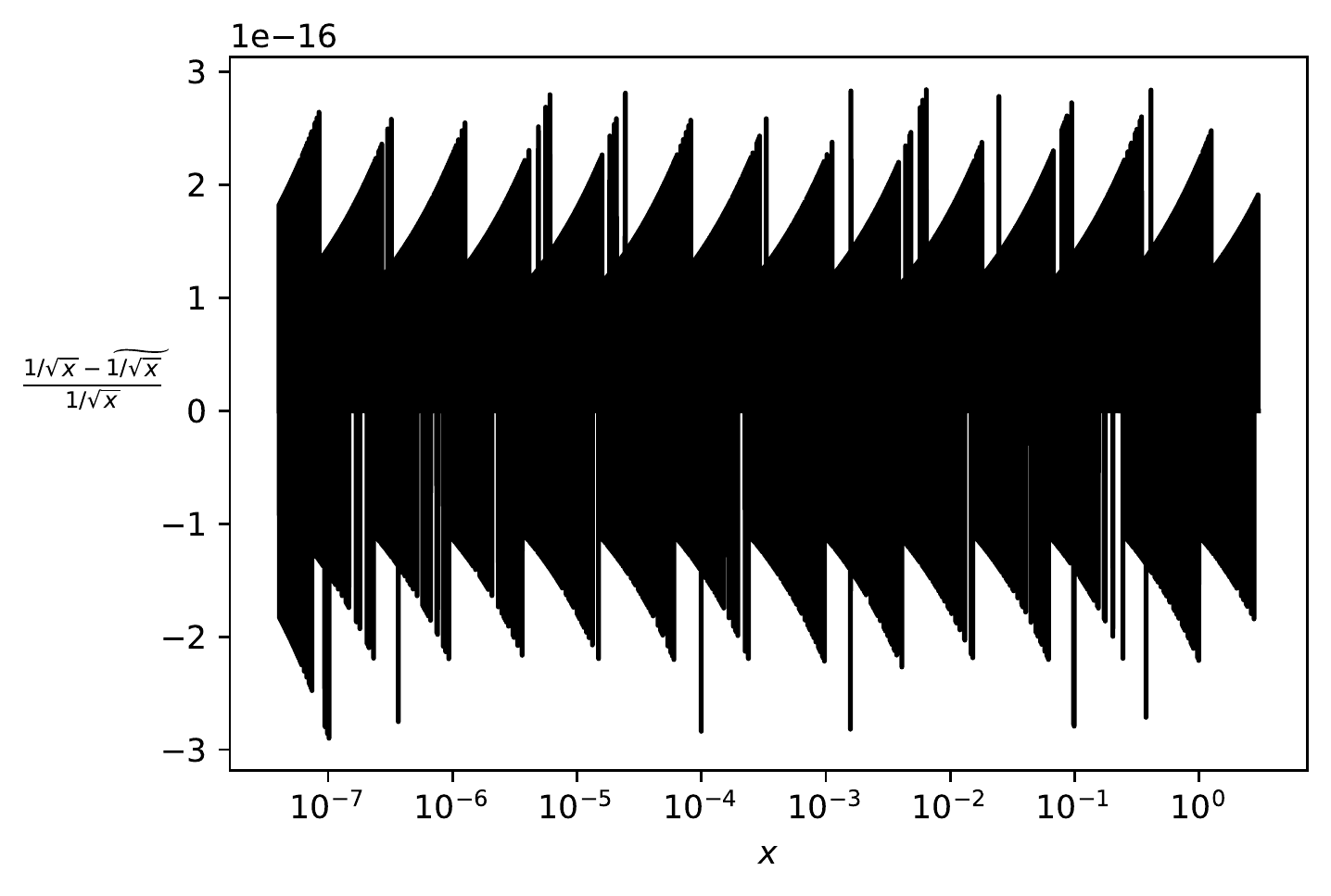}}

\end{centering}
\caption{Relative error in computation of $1/\sqrt{x}$ with 26 iterations
of~(\ref{invs})}
\label{invsqrt}
\end{figure}

\begin{figure}
\begin{centering}

\parbox{.81\textwidth}{\includegraphics[width=.8\textwidth]{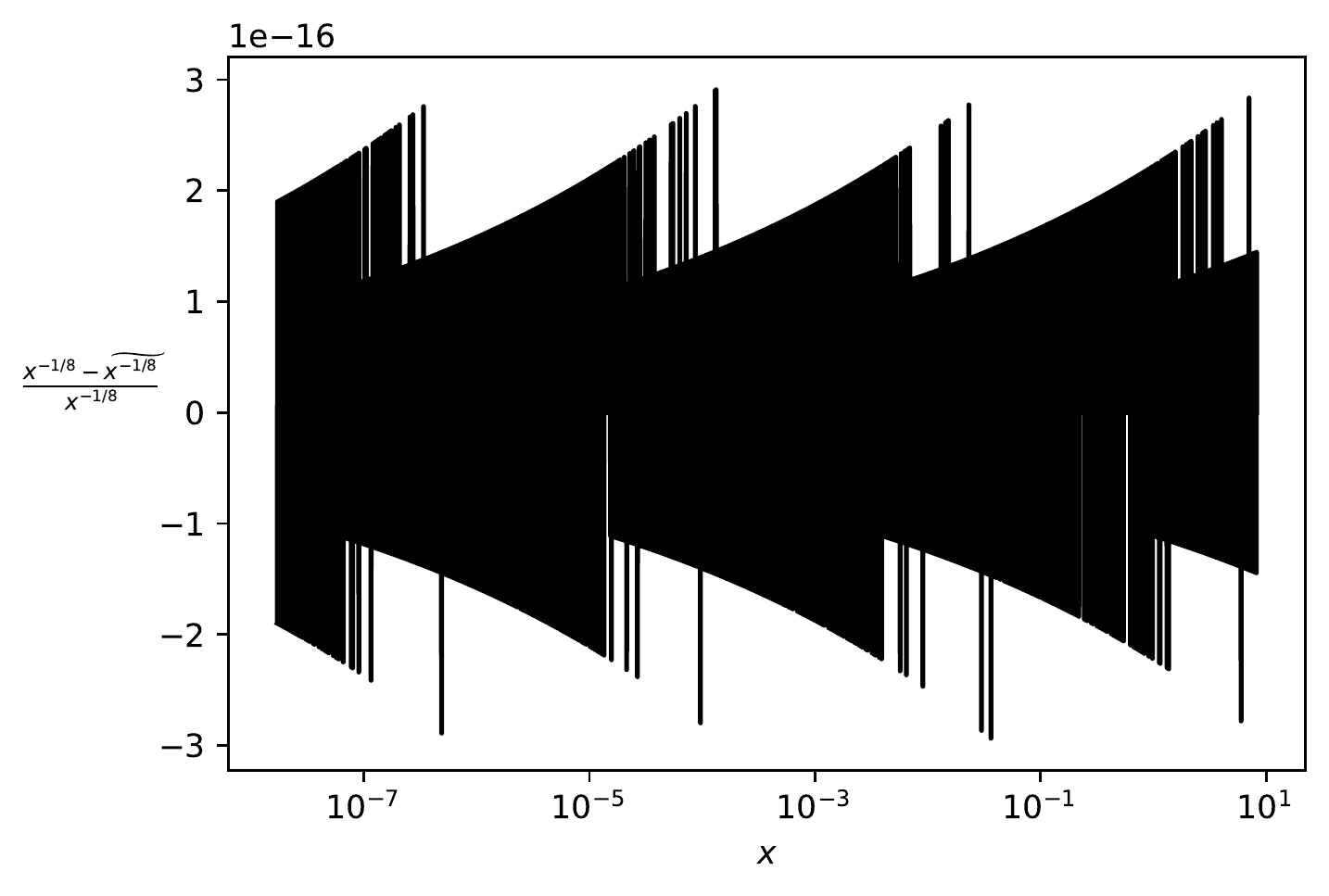}}

\end{centering}
\caption{Relative error in computation of $x^{-1/8}$ with 24 iterations
of~(\ref{inv8})}
\label{inv8root}
\end{figure}

\begin{figure}
\begin{centering}

\parbox{.81\textwidth}{\includegraphics[width=.8\textwidth]{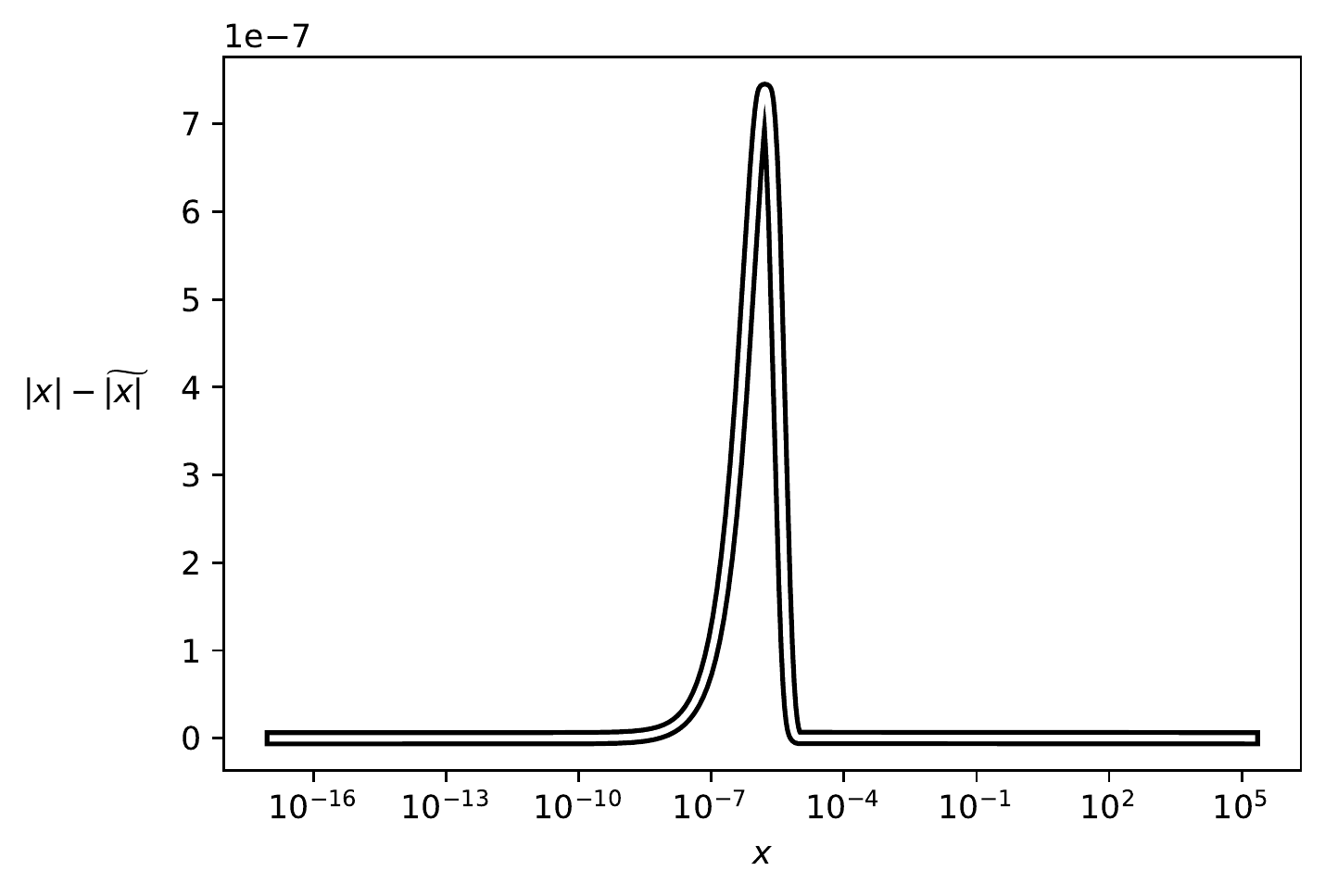}}

\end{centering}
\caption{Absolute error in computation of $|x| = x \sgn(x)$ with 60 iterations
of~(\ref{sgn}); the figure superimposes a white curve over a black curve,
where the white curve uses $y_0 = -x / \gamma$ to start~(\ref{sgn})
while the black curve uses~$y_0 = x / \gamma$, both with $\gamma = 10^5$}
\label{absval}
\end{figure}

\subsection{Chebyshev series}
\label{Chebyshev}

Chebyshev series provide efficient approximations to smooth functions
using only additions and multiplications.
The approximations are especially efficient for odd functions, such as
\begin{equation}
\label{tanh}
f(x) = \tanh(x) = \frac{\exp(x) - \exp(-x)}{\exp(x) + \exp(-x)},
\end{equation}
\begin{equation}
\label{logit}
f(x) = \frac{1}{1 + \exp(-x)} - \frac{1}{2} = \tanh(x/2) / 2,
\end{equation}
and
\begin{equation}
\label{probit}
f(x) = \frac{1}{\sqrt{2\pi}} \int_{-\infty}^x \exp(-y^2/2) \, dy - \frac{1}{2};
\end{equation}
in this subsection, $x$ denotes a real number.
The function in~(\ref{tanh}) is the hyperbolic tangent.
The function in~(\ref{logit}) is a constant
plus the standard logistic function, familiar from logistic regression.
The function in~(\ref{probit}) is a constant
plus the cumulative distribution function for the standard normal distribution,
familiar from probit regression.
Performing logistic regression or probit regression by maximizing
the log-likelihood relies on the evaluation of~(\ref{logit}) or~(\ref{probit}),
respectively, at least when using a gradient-based optimizer,
the method of Newton and Raphson, or the method of scoring.
Details about these functions and regressions are available, for example,
in the monograph of~\cite{mccullagh-nelder}.

Appendix~\ref{chebappendix} reviews algorithms for computing approximations
via Chebyshev series of odd functions,
with accuracy determined via two parameters, $n$ and $z$,
where the degree of the (odd) approximating polynomial is $2n-1$,
and $[-z, z]$ is the interval over which the approximation is valid.
Setting $n = 50$ and $z = 10$ yields 7-digit accuracy
for the approximation of~(\ref{tanh});
setting $n = 22$ and $z = 5$ yields 4-digit accuracy
for the approximation of~(\ref{logit}); and
setting $n = 34$ and $z = 10$ yields 5-digit accuracy
for the approximation of~(\ref{probit}).
In Section~\ref{numex} below, we err on the side of caution,
defaulting to $n = 60$ and $z = 20$ for~(\ref{tanh}) and~(\ref{logit})
and to $n = 50$ and $z = 20$ for~(\ref{probit}),
while also discussing the results from other choices.

\subsection{Softmax}
\label{softmax}

As reviewed, for example, by~\cite{lecun-chopra-hadsell-ranzato-huang},
a common operation in machine learning is the so-called ``softmax''
transforming $n$ non-positive real numbers $x_1$, $x_2$, \dots, $x_n$
into the $n$ positive real numbers
$\exp(x_1)/Z$, $\exp(x_2)/Z$, \dots, $\exp(x_n)/Z$,
where $Z = Z(x_1, x_2, \dots, x_n) = \sum_{k=1}^n \exp(x_k)$.
These are the probabilities
at unit temperature in the Gibbs distribution associated
with energies $-x_1$, $-x_2$, \dots, $-x_n$,
where $Z$ is the partition function.
Once we have computed the exponentials, summation yields $Z$ directly;
the secure multiparty computations of Section~\ref{smpc}
support such summation. Division by $Z$ is available
via the iterations in~(\ref{inverse}) of Subsection~\ref{Newton}.

Thus, given real numbers $x$ and $\beta$ such that $-\beta \le x \le 0$,
and a real number $\epsilon$ such that $0 < \epsilon < 1$,
we would like to calculate $\exp(x)$ to precision $\epsilon$.
We use the method of scaling and squaring, as reviewed, for example,
by~\cite{higham}. If we let $n$ be the least integer
that is at least $\log_2(2\beta^2/\epsilon)$,
then squaring $\exp(x/2^n)$ yields $\exp(x/2^{n-1})$,
squaring $\exp(x/2^{n-1})$ yields $\exp(x/2^{n-2})$, and so on,
so that $n$ successive squarings will yield $\exp(x)$;
further, $1 + x/2^n$ approximates $\exp(x/2^n)$:
\begin{equation}
\left| \exp(x/2^n) - 1 - x/2^n\right|
= \left| \sum_{k=2}^{\infty} (x/2^n)^k/k! \right|
= (x/2^n)^2 \left| \sum_{k=0}^{\infty} (x/2^n)^k/(k+2)! \right|
\le (x/2^n)^2 \exp(x/2^n),
\end{equation}
that is,
\begin{equation}
1 + x/2^n = (1 + \delta) \exp(x/2^n), \quad |\delta| \le (x/2^n)^2,
\end{equation}
while
\begin{equation}
\left|(1 + \delta)^{2^n} - 1\right|
\le \left(1 + (x/2^n)^2\right)^{2^n} - 1
\le \left(\exp((x/2^n)^2)\right)^{2^n} - 1
= \exp(x^2/2^n) - 1 \le 2x^2/2^n \le \epsilon,
\end{equation}
so $n$ successive squarings of $1 + x/2^n$ yields $\exp(x)$
to relative accuracy $\epsilon$ (or better). 
We use $n = 20$ successive squarings in all numerical experiments
of Section~\ref{numex} below.

Given a real number $\gamma > 3\beta$, less than $6n\beta/\gamma$ bits can leak
from computing the approximation to $\exp(x)$
if we add to $1 + x/2^n$ a random variable distributed uniformly
over $[-\gamma/2^n, \gamma/2^n]$, and double the width of the added noise
upon each of the $n$ squarings, in accord with Theorems~\ref{infoleak},
\ref{chaining}, and~\ref{Beaverdetails} of Section~\ref{leakage}.

Clearly, we can enforce that a real number $x$ be non-positive by applying
$-\operatorname{ReLU}(-x)$ from~(\ref{relu}).
Computing the softmax of real numbers that may not necessarily be non-positive
is also possible, even without risk of leaking any information beyond
the case for non-positive numbers.
Indeed, we can replace each real number $x_j$
with $x_j - \sum_{k=1}^n \operatorname{ReLU}(x_k)$,
without altering the probability distribution produced by the softmax;
we perform such replacement in our implementation
of multinomial logistic regression.
When implementing a softmax for multinomial logistic regression,
we include a negative offset in the bias for the input to the softmax.
That is, we adjust the bias to be a constant amount less,
constant over all classes in the classification.
Subtracting such a positive constant $C$ tends
to make $x_1 - C$, $x_2 - C$, \dots, $x_n - C$ negative
even before replacing each real number $x_j - C$
with $x_j - C - \sum_{k=1}^n \operatorname{ReLU}(x_k - C)$.
Subtracting a constant $C$ reduces the sum
$\sum_{k=1}^n \operatorname{ReLU}(x_k - C)$; without subtracting the constant,
the sum $\sum_{k=1}^n \operatorname{ReLU}(x_k)$
can adversely impact accuracy if the sum becomes too large.
Subtracting the constant $C$ reduces accuracy by a factor of up to $\exp(C)$;
we set $C = 5$ (so $\exp(C) \approx 148$ --- a tad more than two digits)
for our numerical experiments reported in the following section,
Section~\ref{numex}.

\section{Numerical examples}
\label{numex}

Via several numerical experiments,
this section illustrates the performance of the scheme proposed above.
All examples reported in the present section use two parties
for the private computations. Subsection~\ref{several} outlines
an extension to several parties.
The terminology ``in private'' refers to computations fully encrypted
via the scheme introduced above,
while ``in public'' refers to computations in plaintext.
Subsection~\ref{synthetic} validates the scheme on examples
for which the correct answer is known by construction.
Subsection~\ref{realdata} applies the scheme to classical data sets.

All examples use minibatched stochastic gradient descent to maximize
the log-likelihood under the corresponding generalized linear model,
at the constant learning rates specified below
(except where noted for probit regression).
Appendix~\ref{sgdrev} briefly reviews stochastic gradient descent
with minibatches and weight decay;
Appendix~\ref{glmrev} briefly reviews generalized linear models.

In all cases, we learn (that is, fit) not only the vector of weights
to which the design matrix gets applied, but also a constant offset
known as the ``bias'' in the literature on stochastic gradient descent.
Thus, the linear function of the weight (fitted parameter) vector $w$
in the generalized linear model is actually the affine transform $Aw + c$,
where $A$ is the design matrix and $c$ is the bias vector whose entries
are all the same constant offset, learned or fitted together with $w$
during the iterations of stochastic gradient descent (whereas $A$ stays fixed
during the iterations).
In multinomial logistic regression, the entries of the bias vector $c$
are constant for each class (constant over all covariates and data samples),
but the constant may be different for different classes.

In the subsections below, we view linear least-squares regression
as a generalized linear model with the link function being the identity;
equivalently, we take the parametric family defining the statistical model
to be an affine transform of the vector of weights (parameters) plus
independent and identically distributed normal random variables.
The log-likelihood of a such a model summed over all samples
in the design matrix $A$ is simply a constant minus $\|Aw+c-b\|_2^2/2$,
where $\| \cdot \|_2$ denotes the Euclidean norm,
$w$ is the vector of weights (parameters),
$Aw+c$ is the affine transform defined by the design matrix $A$
and bias vector $c$, and $b$ is the vector of targets.
For simplicity, when we report the ``negative log-likelihood'' or ``loss''
in plots, averaged over the $m$ samples in the design matrix $A$,
we report $\|Aw+c-b\|_2^2 / m$, ignoring the constant and factor of 2.

The implementation of encrypted computations builds
on CrypTen of~\cite{crypten}, which in turn builds on PyTorch.
The implementation uses only IEEE standard double-precision arithmetic.
All experiments ran on one of Facebook's computer clusters for research,
which enables rapid communication between the multiple parties.
In actual deployments, communications between the multiple parties
are likely to have high latency, dramatically impacting the speed
of the multiparty computations. The speed of such communications
would vary significantly between different applications and arrangements,
likely requiring separate analyses and characterizations
of computational efficiency for different deployments.
Yet, while timings are fairly unique to the particular
computational environment, the accuracies we report below
should be fully representative for most applications.

\subsection{Validations on synthetic data}
\label{synthetic}

For the synthetic examples discussed in the following sub-subsections,
we set $m = 64$ and $n = 8$ for the numbers of rows and columns
in the design matrices being constructed.
Figure~\ref{synthwidths} displays the discrepancy of the computed results
from the ideal solution (the ideal is known a-priori to two-digit accuracy
by construction in these synthetic examples)
as a function of the maximum value $\gamma$ of the random variable
distributed uniformly over $[-\gamma, \gamma]$.
In accordance with the analysis in Subsection~\ref{two-party} above,
the figure reports that $\gamma = 10^5$ works well,
yielding the roughly two-digit agreement that would be optimal
for the synthetic data sets constructed in the following sub-subsections,
so we set $\gamma = 10^5$ for the remainder of the paper.
The logistic and probit regressions reported below both rely
on the Chebyshev approximations reviewed in Subsection~\ref{Chebyshev} above,
with the approximations being valid over the interval $[-20, 20]$
(so $z = 20$ in the notation of Subsection~\ref{Chebyshev}).
Figure~\ref{synthterms} indicates that the degree 40 approximations suffice
for optimal accuracy, while even degrees 6--12 produce reasonably accurate
results (accurate enough for most applications in machine learning
for prediction, in which only residuals or matching targets matters).
For the remainder of the paper, we use 120 terms in the Chebyshev approximation
of the logistic function (or $\tanh$),
and 100 terms in the Chebyshev approximation of the inverse probit
(or the cumulative distribution function of a standard normal variate).
For both Figures~\ref{synthwidths} and~\ref{synthterms},
we trained for 10,000 iterations with 8 samples in each iteration's minibatch,
thus sweeping through a random permutation of the full synthetic data set
1,250 times (each sweep is known as an ``epoch'').
The following sub-subsections detail the construction
of our synthetic data sets.

\begin{figure}
\begin{centering}

\parbox{.71\textwidth}{\includegraphics[width=.7\textwidth]{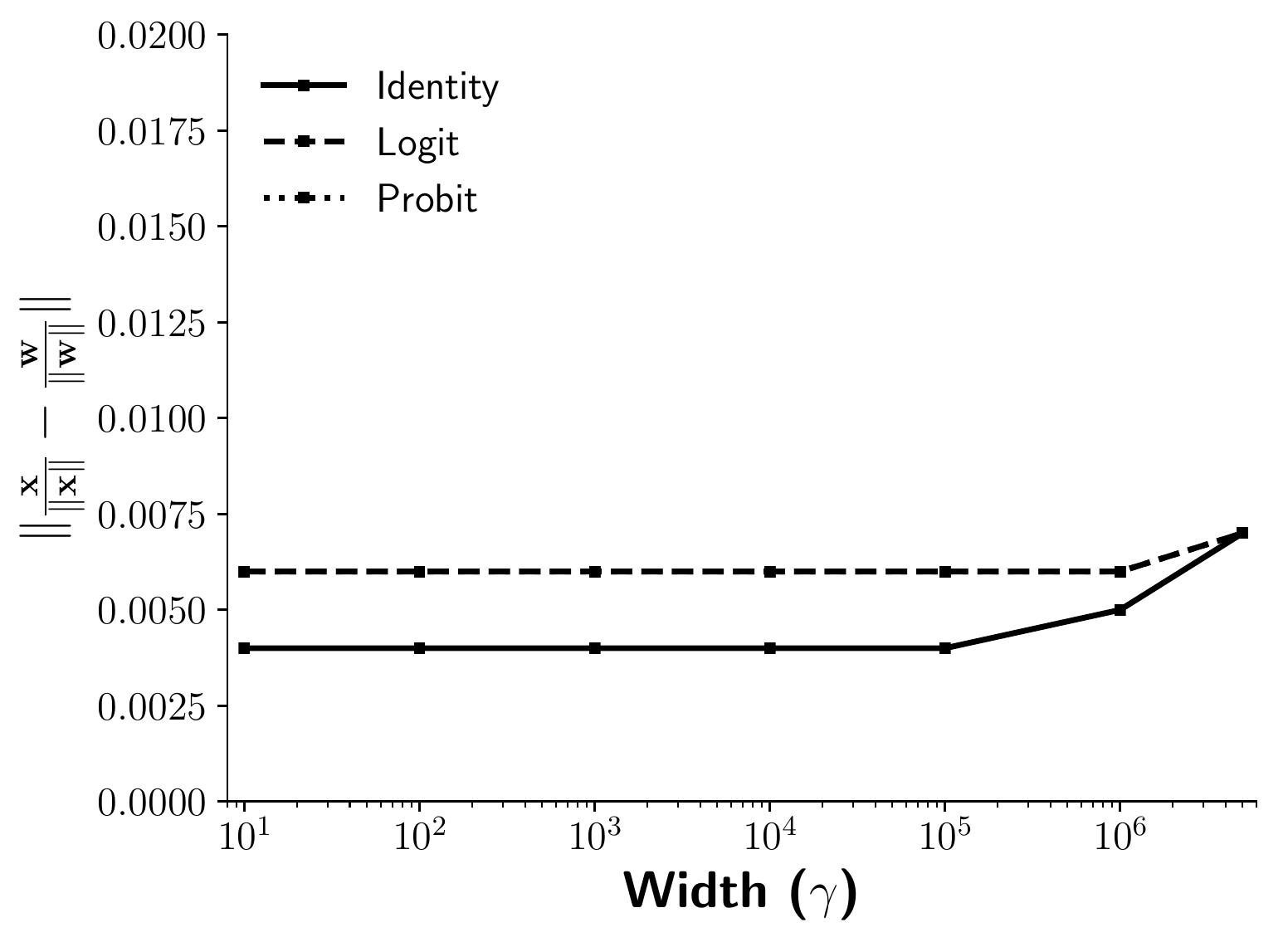}}

\end{centering}
\caption{Euclidean norm of the difference
between the ideal normalized weight vector $w/\|w\|_2$
and its computed approximation $x/\|x\|_2$ as a function of the width $\gamma$
of the uniform noise on $[-\gamma, \gamma]$ added to the shares of data
(the lines for the logit and probit links overlap)}
\label{synthwidths}
\end{figure}

\begin{figure}
\begin{centering}

\parbox{.71\textwidth}{\includegraphics[width=.7\textwidth]{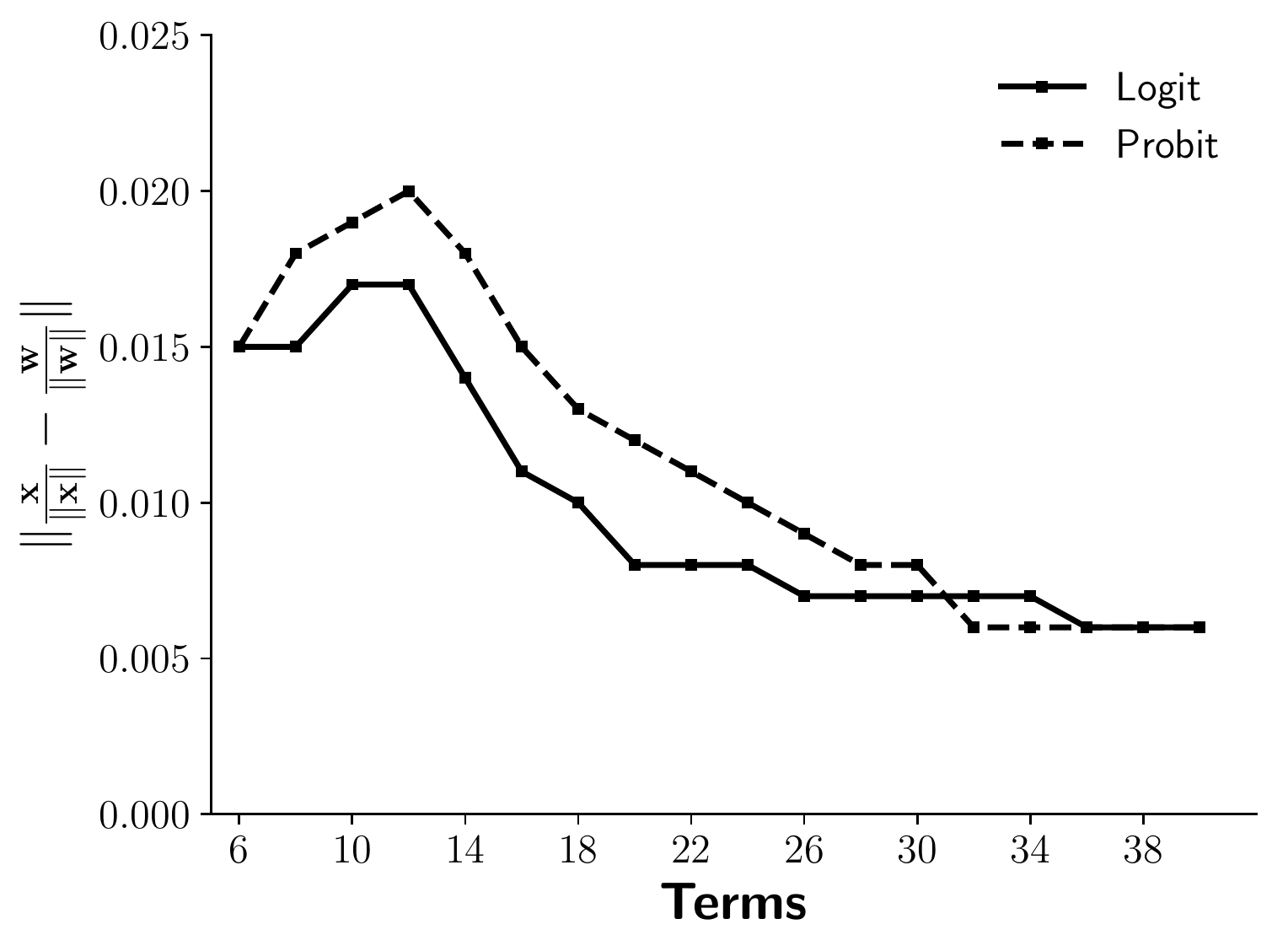}}

\end{centering}
\caption{Euclidean norm of the difference
between the ideal normalized weight vector $w/\|w\|_2$
and its computed approximation $x/\|x\|_2$ as a function of the number of terms
in the Chebyshev approximation of Subsection~\ref{Chebyshev}
(about half the terms vanish, since the polynomial is odd)}
\label{synthterms}
\end{figure}

\subsubsection{Linear least-squares regression (identity link)}

We form the design matrix $A$ and target vector $b$ as follows.
We orthonormalize the columns of an $m \times (n + 1)$ matrix whose entries
are all independent and identically distributed (i.i.d.)\
standard normal variates
to obtain an $m \times n$ matrix $A$ whose columns are orthonormal
($A$ is the leftmost block of $n$ columns)
and an $m \times 1$ column vector $v$ that is orthogonal to all columns of $A$
and such that $\| v \|_2 = 1$ ($v$ is the remaining column).
We define the ideal weights $w$ to be an $n \times 1$ column vector
whose entries are i.i.d.\ standard normal variates.
We define $b$ to be the $m \times 1$ column vector
\begin{equation}
b = Aw + 10 v
\end{equation}
so that by construction
\begin{equation}
\min_x \| Ax - b \|_2^2 = \min_x \| Ax - Aw \|_2^2 + \| 10 v \|_2^2
= \| 10 v \|_2^2 = (10)^2.
\end{equation}
Needless to say, obtaining $x = w$ drives $\| Ax - b \|_2$ to its minimum, 10.

After 10,000 iterations (which is 1,250 epochs) with 8 samples per minibatch
at the learning rate $3 \times 10^{-2}$, the residuals $\| Ax + c - b \|_2$
obtained by training in public and by training in private
are both equal to 10.0 to three significant figures.
For training in private (training in public yields very similar results),
the Euclidean norm of the difference between the ideal $w$
and the computed weight vector $x$ is 0.025
(the Euclidean norm of the difference between $w/\|w\|_2$ and $x/\|x\|_2$
is 0.004), and all entries of the vector $c$ are $-0.008$
(which is $-0.003$ when divided by the Euclidean norm of $x$);
that these values are so small certifies the correctness of the training.

\subsubsection{Logistic regression (logit link)}

We define the ideal weights $w$ to be the result
of normalizing an $n \times 1$ column vector
whose entries are i.i.d.\ standard normal variates, that is, we divide
the vector whose entries are i.i.d.\ by its Euclidean norm, ensuring
\begin{equation}
\label{normalization}
\|w\|_2 = 1.
\end{equation}
The construction of the design matrix $A$ and target vector $t$
is more involved; readers interested only in the results and not the details
of the construction may wish to skip to the last two paragraphs
of the present sub-subsection.

For $j = 1$,~$2$, \dots, $10$,
we construct an $n \times 1$ column vector $v^{(j)}$
whose entries are i.i.d.\ standard normal variates,
project off the component of $v^{(j)}$ along $w$ to obtain $u^{(j)}$,
\begin{equation}
\label{orthogonalization}
u^{(j)} = v^{(j)} - w \sum_{k = 1}^n v^{(j)}_k w_k,
\end{equation}
and set
\begin{equation}
\label{positive}
A_{2j, k} = u^{(j)}_k + 0.02 w_k
\end{equation}
and
\begin{equation}
\label{negative}
A_{2j+1, k} = u^{(j)}_k - 0.02 w_k
\end{equation}
for $k = 1$,~$2$, \dots, $n$.
For the remaining $m - 20$ rows of $A$, we use $m - 20$ rows
from the result of orthonormalizing the columns of an $m \times n$ matrix
whose entries are all i.i.d.\ standard normal variates.

We construct the $m \times 1$ column vector
\begin{equation}
\label{matvec}
b = A w
\end{equation}
and define the target classes
\begin{equation}
\label{targets}
t_j = \operatorname{round}(\sigma(b_j)),
\end{equation}
for $j = 1$,~$2$, \dots, $m$,
where ``$\operatorname{round}$'' rounds to the nearest integer (0 or 1) and
$\sigma$ is the standard logistic function
\begin{equation}
\label{logisticfunc}
\sigma(x) = \frac{1}{1 + \exp(-x)}.
\end{equation}

Combining~(\ref{matvec}), (\ref{positive}), (\ref{negative}),
(\ref{orthogonalization}), and~(\ref{normalization}) yields
\begin{equation}
b_{2j} = \sum_{k=1}^n A_{2j,k} w_k = 0.02
\end{equation}
and
\begin{equation}
b_{2j+1} = \sum_{k=1}^n A_{2j+1,k} w_k = -0.02
\end{equation}
for $j = 1$,~$2$, \dots, $10$. Since, for $j = 1$,~$2$, \dots, $10$,
$b_{2j}$ is slightly positive while $b_{2j+1}$ is slightly negative,
the target classes $t_{2j}$ and $t_{2j+1}$ defined in~(\ref{targets})
will be 1 and 0, respectively,
even though the difference between the corresponding $(2j)$th and $(2j+1)$th
rows of $A$ is small --- combining~(\ref{positive}), (\ref{negative}),
and~(\ref{normalization}) yields that the Euclidean norm of their difference
is 0.04. The decision hyperplane separating class 0 from class 1 will thus
have to pass between 10 pairs of points in $n$-dimensional space ($n = 8$),
with the points in each pair very close to each other
(albeit on opposite sides of the decision hyperplane).
Classifying all these points correctly hence determines the hyperplane
to reasonably high accuracy.

Needless to say, obtaining $x = w$ and $c = 0$
produces perfect accuracy for the logistic regression
which classifies by rounding the result of~(\ref{logisticfunc})
applied to each entry of $Ax + c$, as then $Ax = Aw = b$,
and the target classes in $t$ are the result of~(\ref{logisticfunc})
applied to each entry of $b$. In fact,
obtaining $x$ as any positive multiple of $w$ together with $c = 0$
yields perfect accuracy --- any positive multiple of a vector orthogonal
to the hyperplane separating the two classes specifies that same hyperplane.

After 10,000 iterations (which is 1,250 epochs)
with 8 samples per minibatch at the learning rate 3,
training binary logistic regression in private
(training in public produces very similar results)
drives the Euclidean norm of the difference between $x/\|x\|_2$
and the ideal $w/\|w\|_2$ to 0.006 and drives every entry
of $c / \|x\|_2$ to 0.006,
where $c$ is the vector of offsets (whose entries are all the same).
That these values are so small validates the training.
The log-likelihood, averaged over all $m = 64$ samples,
changes from $-0.785$ to $-0.088$ over the 10,000 iterations
(needless to say, the log-likelihood cannot exceed 0).
The accuracy becomes exactly perfect (that is, 1).

When training multinomial logistic regression for two classes
on the same synthetic data set with the same settings,
similar validation attains:
for training in private (training in public yields very similar results),
the Euclidean norm of the difference between the ideal $w/\|w\|_2$
and the computed $x/\|x\|_2$ becomes $0.008$,
and every entry of $c/\|x\|_2$ becomes either $-0.037$ or $-0.033$,
where $c$ contains the bias offsets.
The log-likelihood, averaged over all $m = 64$ samples,
changes from $-0.958$ to $-0.047$ over the 10,000 iterations,
and the accuracy becomes perfect (that is, exactly 1).
Of course, using multinomial logistic regression with a binomial distribution
rather than standard logistic regression makes no sense,
but these results certify the correctness
of the multinomial logistic regression nonetheless.

\subsubsection{Probit regression (probit link)}
\label{probitlink}

The design matrix $A$ and target vector $t$ for probit regression
are the same as for logistic regression, but replacing the sigmoid $\sigma$
defined in~(\ref{logisticfunc}) with the inverse probit
\begin{equation}
\label{inverseprobit}
\sigma(x) = \frac{1}{\sqrt{2\pi}} \int_{-\infty}^x \exp(-y^2/2) \, dy,
\end{equation}
which is also known as the cumulative distribution function
for standard normal variates.
An iteration of stochastic gradient descent involves selecting rows
of the design matrix $A$ and collecting them together into a matrix $R$,
as well as collecting together the corresponding targets from $t$
into a vector $s$  (the number of rows is the size of the minibatch).
One iteration of stochastic gradient descent for maximizing
the log-likelihood in logistic regression 
updates the weight vector $x$ by adding to $x$ the learning rate $\eta$
times the transpose $R^\top$ applied to the difference
between the corresponding target samples $s$
and $\sigma$ defined in~(\ref{logisticfunc}) applied to each entry of $Rx+c$;
with a minor abuse of notation, we could write that $x$ updates
to $x + \eta R^\top (s-\sigma(Rx+c))$.
Since the sigmoid defined in~(\ref{inverseprobit}) is numerically very similar
to the sigmoid defined in~(\ref{logisticfunc})
(after scaling such that the variances of the sigmoids are the same),
we use the same updating formula for probit regression
as for logistic regression, but using the design matrix $A$
associated with probit regression rather than that for logistic regression,
and using the sigmoid $\sigma$ associated with probit regression
rather than that for logistic regression.
A na\"ive application of stochastic gradient descent for maximizing
the log-likelihood in probit regression would update the weight vector
in the same direction, but scaled slightly, effectively altering
the learning rate a tiny bit from iteration to iteration;
we omit the extra computations required to follow the na\"ive method exactly.

As with logistic regression,
obtaining $x = w$ and $c = 0$ produces perfect accuracy
for the probit regression which classifies by rounding
the result of~(\ref{inverseprobit}) applied to each entry of $Ax + c$. And,
again, obtaining $x$ as any positive multiple of $w$ together with $c = 0$
yields perfect accuracy --- any positive multiple of a vector orthogonal
to the hyperplane separating the two classes specifies that same hyperplane.
After 10,000 iterations (which is 1,250 epochs)
with 8 samples per minibatch at the learning rate 3,
training in private (training in public produces very similar results)
drives the Euclidean norm of the difference between $x/\|x\|_2$
and the ideal $w/\|w\|_2$ to 0.006 and drives every entry
of $c / \|x\|_2$ to 0.005,
where $c$ is the vector of offsets (whose entries are all the same).
That these values are so small validates the training.
The log-likelihood, averaged over all $m = 64$ samples,
changes from $-0.918$ to $-0.058$ over the 10,000 iterations
(needless to say, the log-likelihood cannot exceed 0).
The accuracy becomes precisely perfect (that is, 1).

\subsubsection{Poisson regression (log link)}

We obtain the design matrix $A$ by orthonormalizing the columns
of an $m \times n$ matrix whose entries are i.i.d.\ standard normal variates.
We define $w$ to be 10 times the result of normalizing
an $n \times 1$ column vector whose entries
are i.i.d.\ standard normal variates, that is, we divide
the vector whose entries are i.i.d.\ by its Euclidean norm,
and then multiply by 10, ensuring
\begin{equation}
\label{bigten}
\|w\|_2 = 10.
\end{equation}
We construct the $m \times 1$ column vector
\begin{equation}
\label{biased}
b = Aw + 3,
\end{equation}
where ``3'' indicates the $m \times 1$ column vector whose entries are all 3.
We then define the target counts
\begin{equation}
\label{inttargets}
t_j = \operatorname{round}(\exp(b_j)),
\end{equation}
for $j = 1$,~$2$, \dots, $m$,
where ``$\operatorname{round}$'' rounds to the nearest integer
(0, 1, 2, \dots).
Using 10 as the norm of $w$ in~(\ref{bigten}) ensures that the integers
$t_1$,~$t_2$, \dots, $t_m$ vary over a significant range,
while using 3 in the right-hand side of~(\ref{biased}) ensures that,
on average, half will be greater than $\exp(3) \approx 20$.
Having targets that vary over a significant range and with many not too small
ensures that the maximum-likelihood estimates of $w$ and 3
in Poisson regression are close to $w$ and 3 with reasonably high accuracy ---
the discretization from the rounding operation in~(\ref{inttargets})
matters little when many counts are large and spread over a range
significantly greater than the discretization spacing.

After 10,000 iterations (which is 1,250 epochs)
with 8 samples per minibatch at the learning rate $3 \times 10^{-3}$,
training in private (training in public produces very similar results)
drives the Euclidean norm of the difference between $x/\|x\|_2$
and the ideal $w/\|w\|_2$ to 0.003 and drives every entry of $c-3$ to 0.003,
where $c$ is the vector of offsets (whose entries are all the same).
That these values are so small certifies the correctness of the training,
as does the following result:
the log-likelihood for the obtained weight vector $x$ and offset $c$,
averaged over all $m = 64$ samples, is $-2.349$ after the 10,000 iterations,
which matches the log-likelihood for the ideal weight vector $w$
and ideal offset (3) to three-digit precision.

\subsection{Performance on measured data}
\label{realdata}

The following sub-subsections illustrate the application
of the scheme proposed above to several benchmark data sets,
namely, handwritten digits from MNIST, forest covers from covtype,
and the numbers of deaths from horsekicks
in corps of the Prussian army over two decades.

\subsubsection{MNIST}

We use both binary and multinomial logistic regression,
as well as probit regression,
to analyze a classic data set of handwritten digits,
created by Yann LeCun, Corinna Cortes, and Christopher J. C. Burges
via merging two sets from the National Institute of Standards and Technology;
the mixed NIST set is available at \url{http://yann.lecun.com/exdb/mnist}
as a training set of 60,000 samples and a testing set of 10,000 examples.
Each sample is a centered 28-pixel $\times$ 28-pixel grayscale image
of one of the digits 0--9, together with a label for which one;
pixel values can range from 0 to 1.
For multinomial logistic regression we use all 10 classes
(with one class per digit); for binary logistic and probit regressions,
we use the 2 classes corresponding to the digits 0 and 1,
for which there are 12665 samples in the training set,
and 2115 samples in the testing set.
We set the number of samples in a minibatch to 50 for training
with all 10 classes, and to 85 for training with only the 2 classes
corresponding to the digits 0 and 1.
We trained for 20 epochs (that is, 24,000 iterations)
at learning rate $10^{-2}$ with all 10 classes,
and for 30 epochs (that is, 4,470 iterations) at learning rate $10^{-3}$
with only the 2 classes corresponding to the digits 0 and 1.
During training for all 10 classes, we supplemented each iteration
of stochastic gradient descent with weight decay of $10^{-3}$,
which is equivalent to adding to the objective function being minimized
(that is, to the negative log-likelihood)
a regularization term of $10^{-3}$ times half the square
of the Euclidean norm of the weights.
This weight decay has negligible impact on accuracy
yet ensures that the constant 5 we subtract from the bias
in the stochastic gradient descent is sufficient to make almost all inputs
of the softmax be non-positive, as suggested in the last paragraph
of Subsection~\ref{softmax} above.

Figure~\ref{mnisttrain} plots the results of training
and Figure~\ref{mnisttest} displays the performance
of the resulting trained model when applied to the testing set,
both as a function of the maximum value $\gamma$ of the random variable
distributed uniformly over $[-\gamma, \gamma]$;
the figures show that $\gamma = 10^5$ works well, in agreement
with the analysis in Subsection~\ref{two-party} above.
Table~\ref{mnisttab} details the results for $\gamma = 10^5$;
training in public produces the same results
at the precision reported in the table.
For this application to machine learning,
even $\gamma = 10^6$ produces practically perfect predictions
during both training and testing.
Logistic regression corresponds to the logit link;
probit regression corresponds to the probit link.
The value of the log-likelihood averaged over the testing set
is remarkably similar to the average value over the training set,
showing that training generalizes well to the independent testing set.

Figure~\ref{mnistaccs} displays the results of training
multinomial logistic regression for all 10 classes (one class per digit),
along with applying the resulting trained model to the testing set,
as a function of the maximum $\gamma$ of the random variable
distributed uniformly over $[-\gamma, \gamma]$;
the accuracy exceeds 0.9\ from $\gamma = 10^3$ to $\gamma = 10^6$.
The generalization from the training set to the testing set is ideal.
At $\gamma = 10^5$, the training loss is 0.318, the testing loss is 0.305,
the training accuracy is 0.913, and the testing accuracy is 0.917;
training in public yields the same results to three-digit precision.

\begin{table}
\begin{center}
\begin{tabular}{lllll}
         &              train &               test &    train &     test \\
    link & loss\phantom{racy} & loss\phantom{racy} & accuracy & accuracy \\
\hline
identity &              0.063 &              0.057 &          &          \\
   logit &              0.039 &              0.033 &    0.997 &    0.999 \\
  probit &              0.025 &              0.019 &    0.997 &    0.999
\end{tabular}
\end{center}
\caption{Values of the negative log-likelihood (the ``loss'')
and accuracy averaged over the training samples or testing samples
from MNIST,
with $\gamma = 10^5$ being the maximal possible value of the random variable
distributed uniformly over $[-\gamma, \gamma]$ added to shares of the data}
\label{mnisttab}
\end{table}

\begin{figure}
\begin{centering}

\parbox{.71\textwidth}{\includegraphics[width=.7\textwidth]{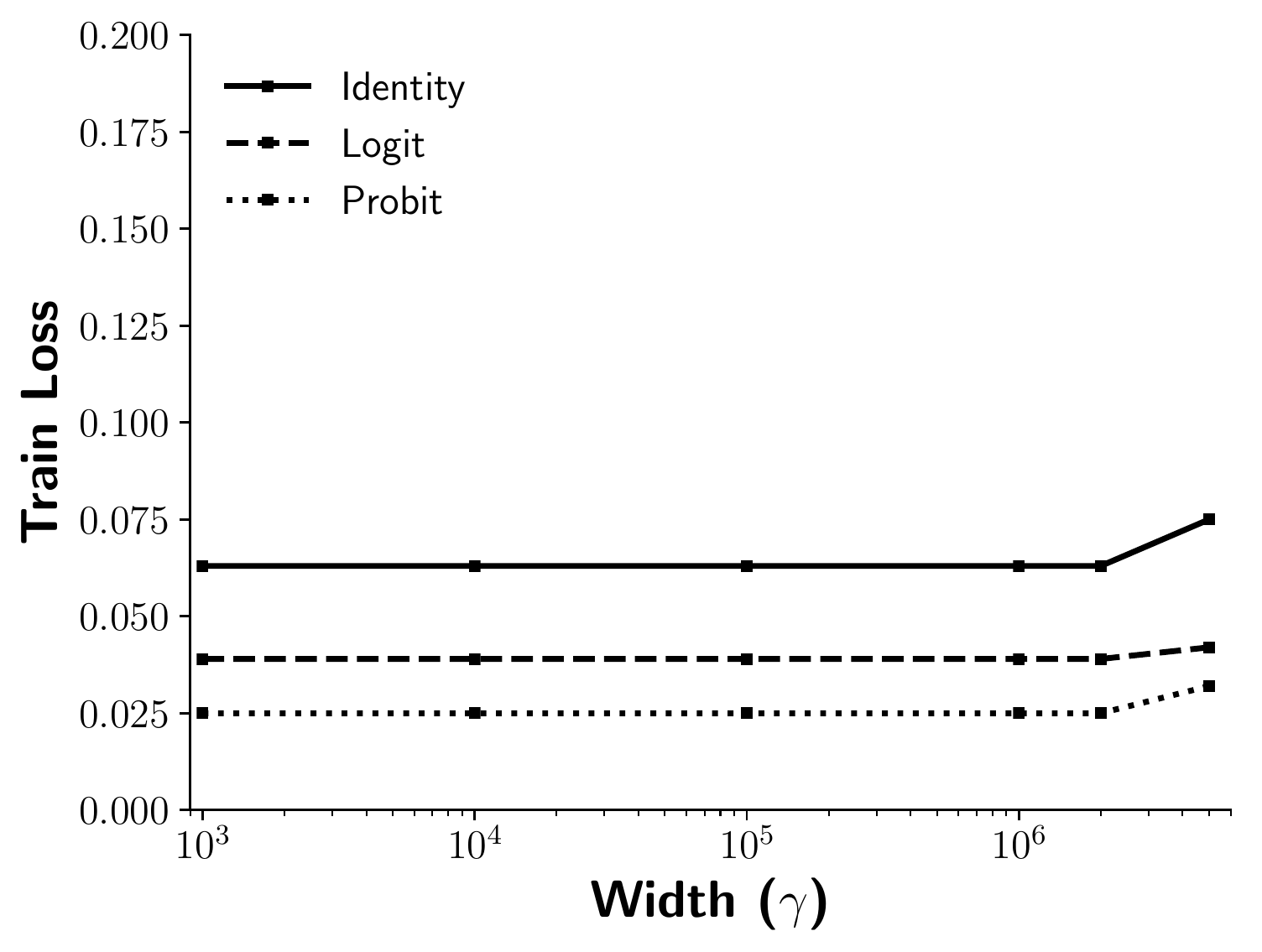}}

\end{centering}
\caption{Value of the negative log-likelihood (the ``loss'')
averaged over the training data from MNIST,
after convergence of the training iterations, as a function
of the maximum value $\gamma$ of the random variable distributed uniformly
on $[-\gamma, \gamma]$ added to the shares of the data}
\label{mnisttrain}
\end{figure}

\begin{figure}
\begin{centering}

\parbox{.71\textwidth}{\includegraphics[width=.7\textwidth]{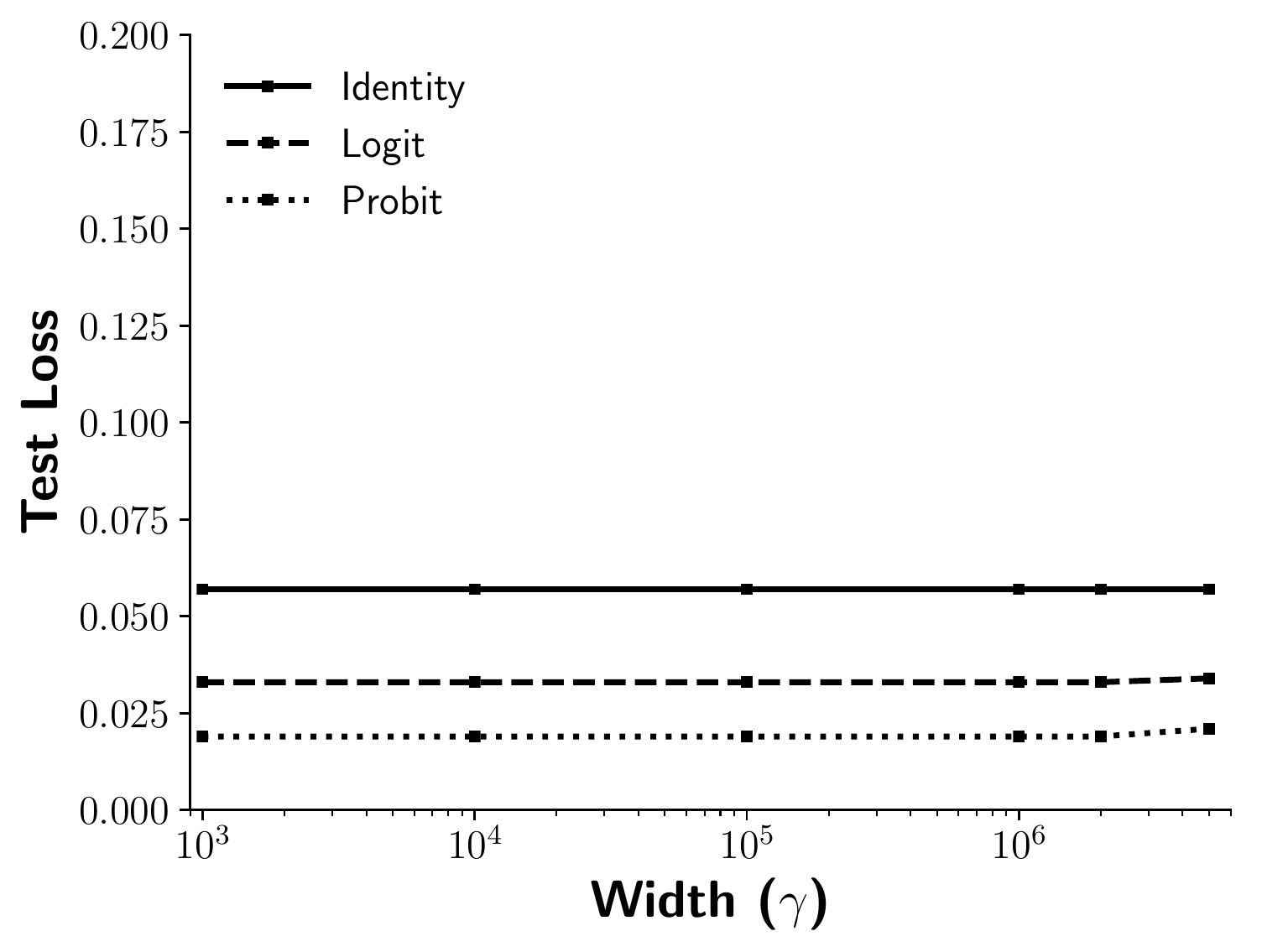}}

\end{centering}
\caption{Value of the negative log-likelihood (the ``loss'')
averaged over the testing data from MNIST, as a function
of the maximum value $\gamma$ of the random variable distributed uniformly
on $[-\gamma, \gamma]$ added to the shares of the data}
\label{mnisttest}
\end{figure}

\begin{figure}
\begin{centering}

\parbox{.71\textwidth}{\includegraphics[width=.7\textwidth]{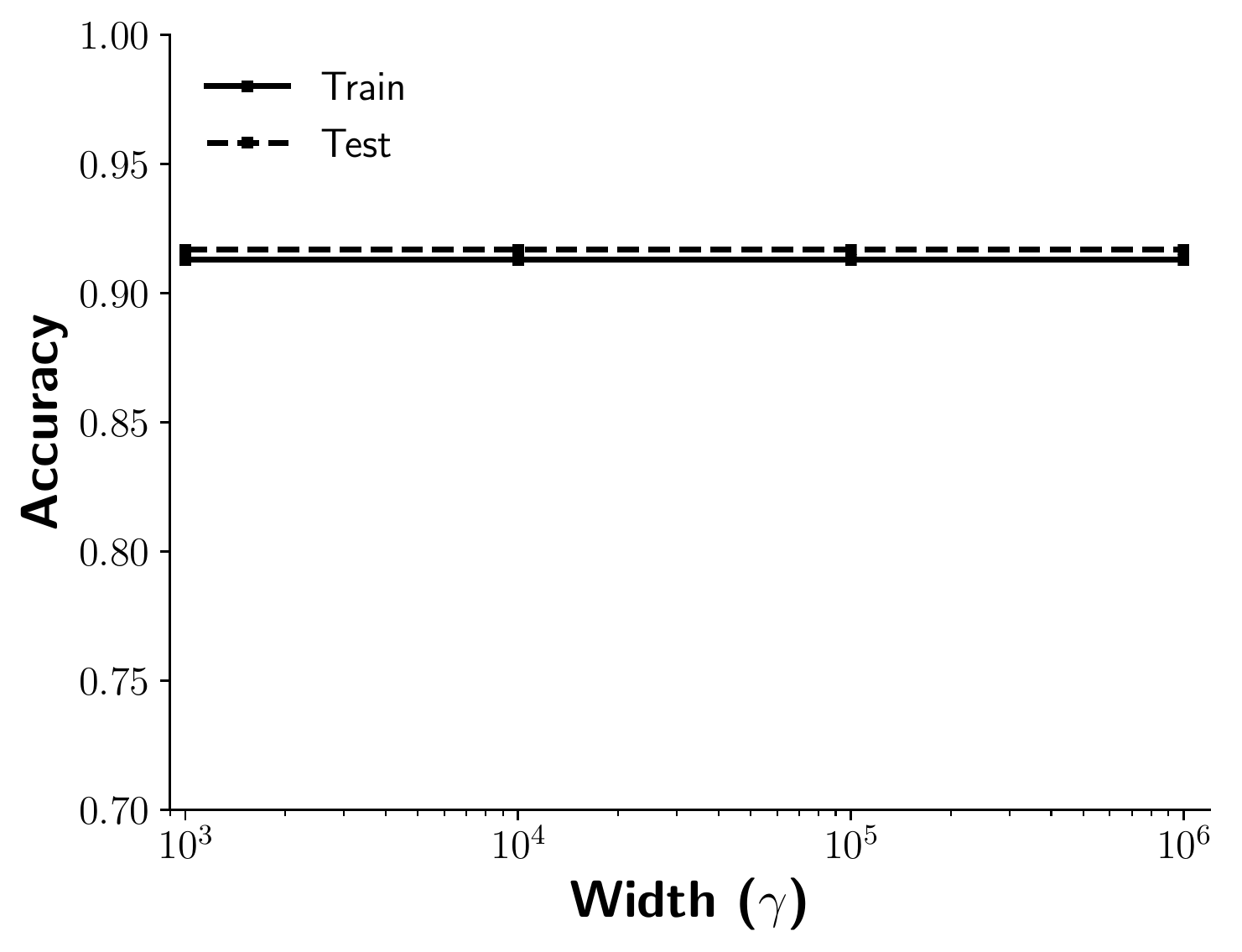}}

\end{centering}
\caption{Accuracy of multinomial logistic regression averaged
over the training or testing data from MNIST,
as a function of the maximum value $\gamma$
of the random variable distributed uniformly on $[-\gamma, \gamma]$
added to the shares of the data}
\label{mnistaccs}
\end{figure}

\subsubsection{Forest cover type}

We use both binary and multinomial logistic regression,
as well as probit regression,
to analyze standard data on the type of forest cover
based on cartographic variates from Jock A. Blackard
of the United States Forest Service,
Denis J. Dean of the University of Texas at Dallas,
and Charles W. Anderson of Colorado State University
(with copyright retained by Jock A. Blackard and Colorado State University);
the original data is available from the archive of~\cite{dua-graff} at
\url{http://archive.ics.uci.edu/ml/datasets/covertype}
and the preprocessed and formatted versions that we use are available
from the work of~\cite{chang-lin} at
\url{http://www.csie.ntu.edu.tw/~cjlin/libsvmtools/datasets/binary.html#covtype.binary} (for binary classification)
and
\url{http://www.csie.ntu.edu.tw/~cjlin/libsvmtools/datasets/multiclass.html#covtype} (for the multinomial logistic regression).

We predict one of 7 types of forest (or one of 2 for the binarized data)
based on 10 integer-valued covariates
(elevation, aspect, slope, horizontal and vertical distances to bodies
of water, horizontal distance to roadways, horizontal distance to fire points,
and hillshade at 9am, 12pm, and 3pm), as well as one-hot encodings
of 4 types of wilderness areas and 40 types of soil.
Thus, there are 54 covariates in all, including the one-hot encodings.
(A one-hot encoding is a vector whose entries are all 0 except for one 1
in the position corresponding to the associated type.)
For normalization, we subtract the mean
from each of the integer-valued covariates (not from the one-hot encodings)
and then divide by the maximum absolute value.
We randomly permute and then partition the data
into a training set of 500,000 samples and a testing set
of the other 81,012 samples.
We trained for 20 epochs (that is, 10,000 iterations)
with 1,000 samples per minibatch.
The learning rate for the binary classification was 3
(for identity link was 0.1)
and for the multi-class (7-class) classification was 1.
During training for all 7 classes, we supplemented each iteration
of stochastic gradient descent with weight decay of $10^{-3}$,
to be consistent with our training for all 10 classes of MNIST
in the previous sub-subsection.
This weight decay barely impacts accuracy
yet makes the constant 5 that we subtract from the bias
in the stochastic gradient descent shift almost all inputs
of the softmax to be non-positive, as suggested in the last paragraph
of Subsection~\ref{softmax} above.

Figure~\ref{covtypetrain} displays the results of training
and Figure~\ref{covtypetest} depicts the performance
of the resulting trained model when applied to the testing set,
both as a function of the width $\gamma$ of the random variable
distributed uniformly over $[-\gamma, \gamma]$;
the figures show that $\gamma = 10^5$ works well, in accordance
with the analysis in Subsection~\ref{two-party} above.
Table~\ref{covtypetab} details the results for $\gamma = 10^5$;
training in public yields the same results
to within $\pm 0.001$ of those reported in the table.
In fact, even $\gamma = 10^6$ works perfectly fine for this application
to machine learning.
Logistic regression corresponds to the logit link;
probit regression corresponds to the probit link.
The value of the log-likelihood averaged over the testing set
is reassuringly close to the average value over the training set,
demonstrating good generalization from the training set to the testing set.

Figure~\ref{covtypeaccs} displays the results of training
multinomial logistic regression for all 7 classes,
together with applying the resulting trained model to the testing set,
as a function of the maximum $\gamma$ of the random variable
distributed uniformly over $[-\gamma, \gamma]$;
the accuracy is excellent from $\gamma = 10^3$ to $\gamma = 10^6$.
The generalization from the training set to the testing set is perfect.
At $\gamma = 10^5$, the training loss is 0.705, the testing loss is 0.711,
the training accuracy is 0.710, and the testing accuracy is 0.710;
training in public produces the same results to within $\pm 0.001$.

\begin{table}
\begin{center}
\begin{tabular}{lllll}
         &              train &               test &    train &     test \\
    link & loss\phantom{racy} & loss\phantom{racy} & accuracy & accuracy \\
\hline
identity &              0.175 &              0.176 &          &          \\
   logit &              0.515 &              0.515 &    0.755 &    0.758 \\
  probit &              0.517 &              0.517 &    0.755 &    0.756
\end{tabular}
\end{center}
\caption{Values of the negative log-likelihood (the ``loss'')
and accuracy averaged over the training samples or testing samples
from data on forest cover type,
with $\gamma = 10^5$ being the maximal possible value of the random variable
distributed uniformly over $[-\gamma, \gamma]$ added to shares of the data}
\label{covtypetab}
\end{table}

\begin{figure}
\begin{centering}

\parbox{.71\textwidth}{\includegraphics[width=.7\textwidth]{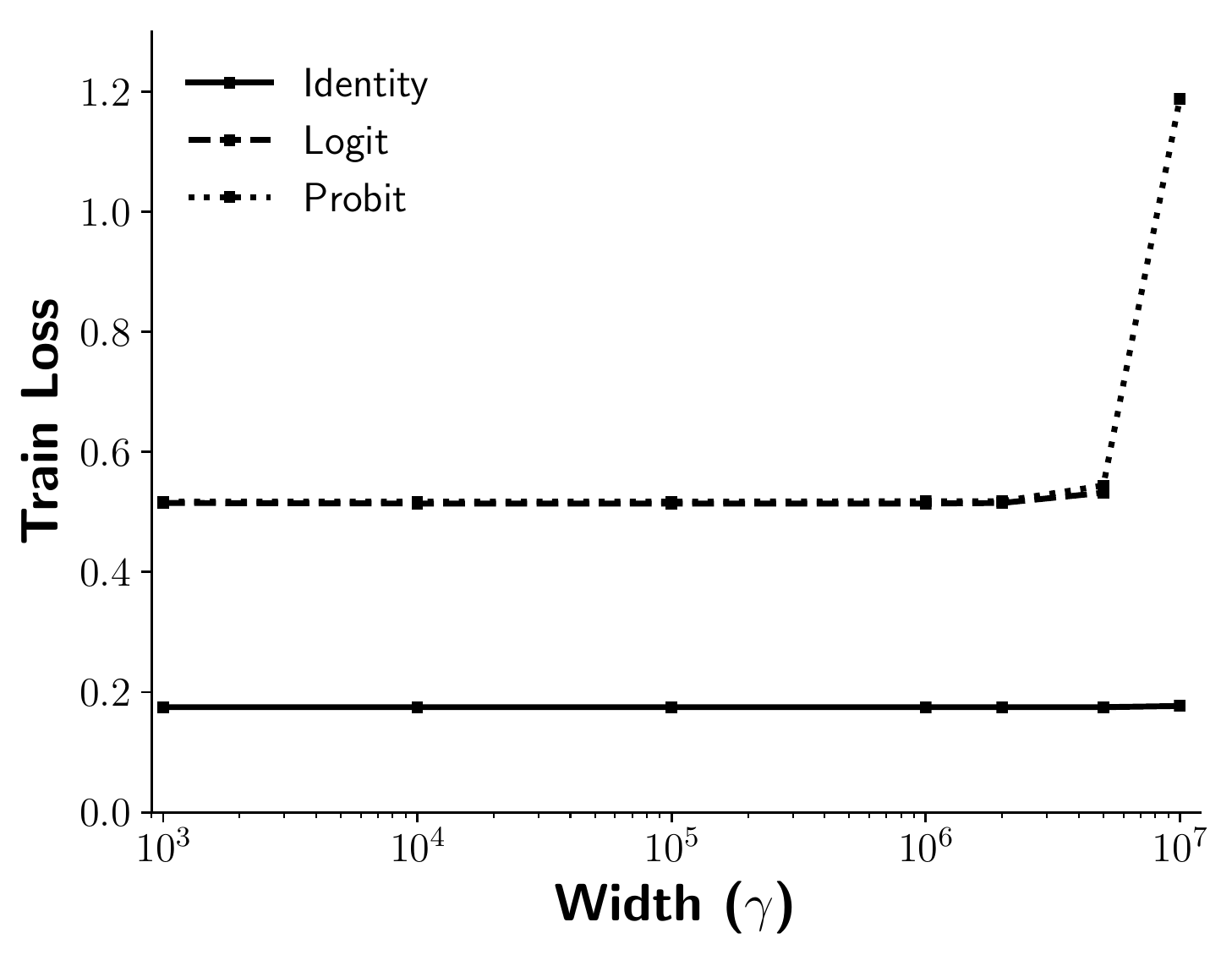}}

\end{centering}
\caption{Value of the negative log-likelihood (the ``loss'')
averaged over the training data on forest cover type,
after convergence of the training iterations, as a function
of the maximum value $\gamma$ of the random variable distributed uniformly
on $[-\gamma, \gamma]$ added to the shares of the data}
\label{covtypetrain}
\end{figure}

\begin{figure}
\begin{centering}

\parbox{.71\textwidth}{\includegraphics[width=.7\textwidth]{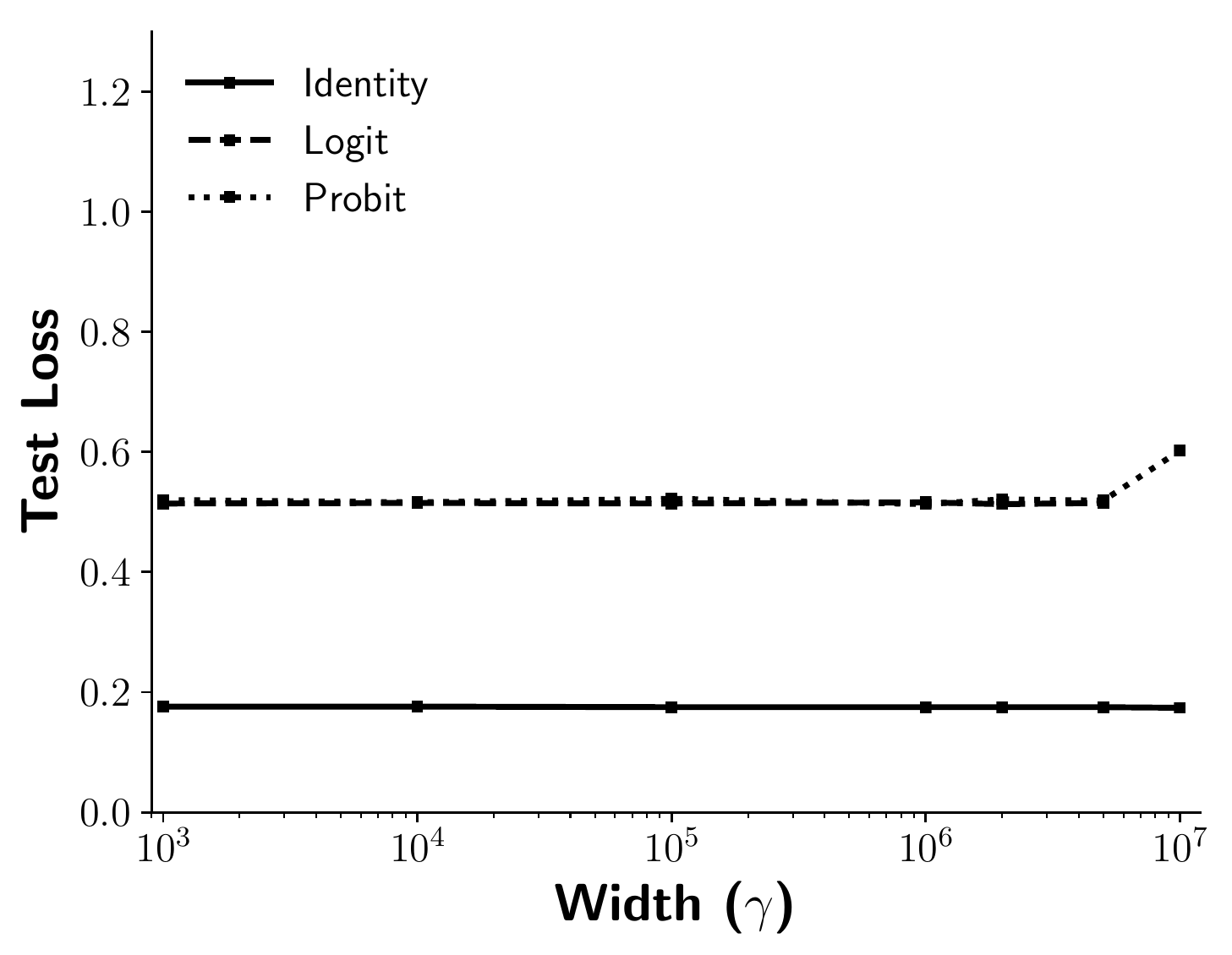}}

\end{centering}
\caption{Value of the negative log-likelihood (the ``loss'')
averaged over the testing data on forest cover type, as a function
of the maximum value $\gamma$ of the random variable distributed uniformly
on $[-\gamma, \gamma]$ added to the shares of the data}
\label{covtypetest}
\end{figure}

\begin{figure}
\begin{centering}

\parbox{.71\textwidth}{\includegraphics[width=.7\textwidth]{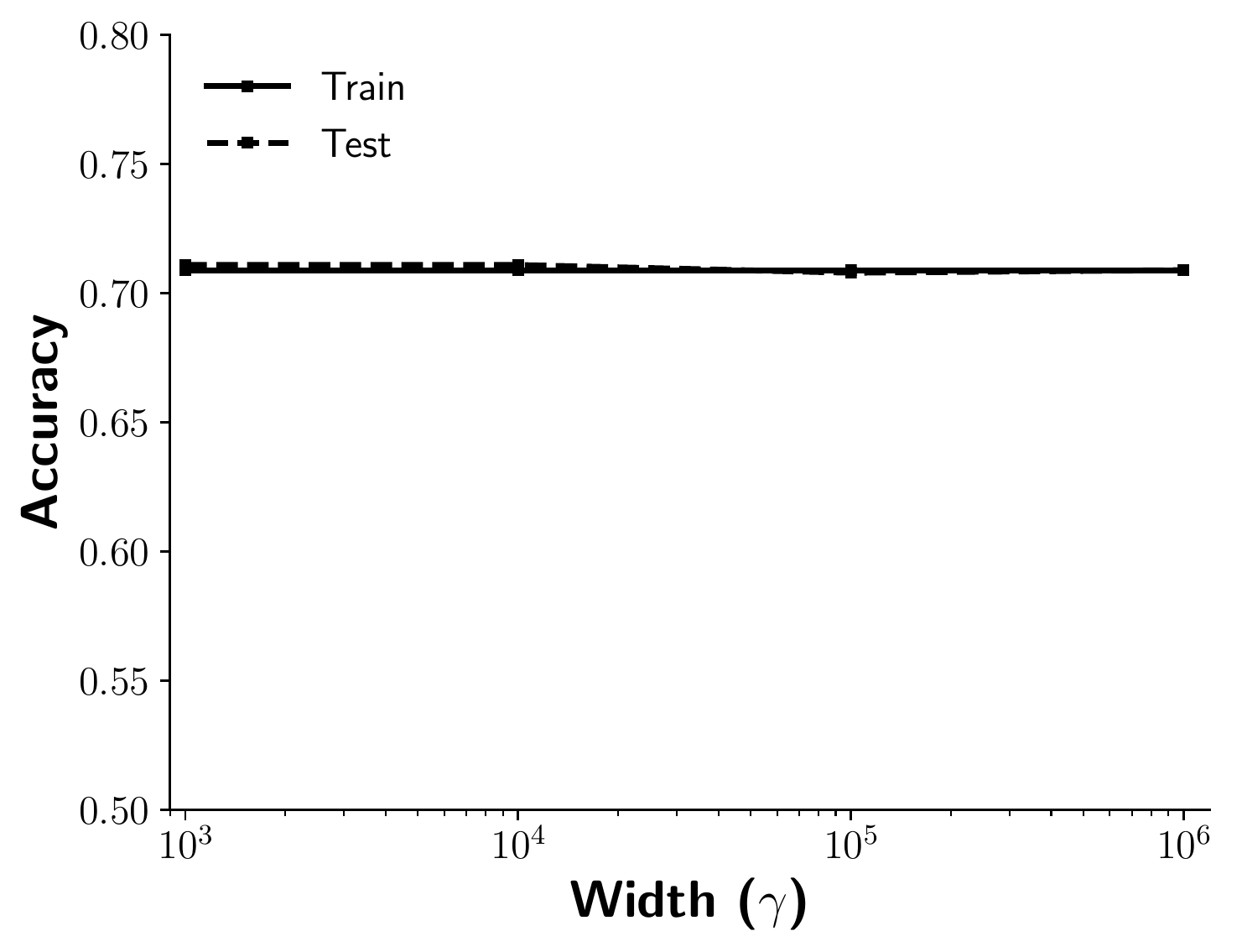}}

\end{centering}
\caption{Accuracy of multinomial logistic regression averaged
over the training or testing data on forest cover type,
as a function of the maximum value $\gamma$
of the random variable distributed uniformly on $[-\gamma, \gamma]$
added to the shares of the data}
\label{covtypeaccs}
\end{figure}

\subsubsection{Deaths from horsekicks}

We use Poisson regression to analyze the classical data
from Ladislaus Bortkiewicz tabulating the numbers of deaths from horsekicks
in 14 corps of the Prussian army for each of the 20 years from 1875 to 1894,
available (with extensive discussion) in the work of~\cite{stigler}.
This data is a canonical example of counts which follow
the Poisson distribution.
We consider four separate Poisson regressions,
for the following sets of covariates:
(0) no covariates, (1) one-hot encodings of the corps,
(2) second-order polynomials of the years, and
(3) concatenating the one-hot encodings of the corps
and the second-order polynomials of the years.
The one-hot encoding of a corps is a vector with 14 entries,
one of which is 1 and 13 of which are 0; the position of the entry that is 1
corresponds to the associated corps.
The second-order polynomials of the years arise from using as covariate vector
a vector with 3 entries, the entries being the constant 1,
the number of years beyond 1875,
and the square of the number of years beyond 1875; Poisson regression
considers linear combinations of these entries,
thus forming quadratic functions of the years.

The log-likelihoods of the fully trained models are the same
to three-digit precision when comparing training in public
to training in private; with 14 samples per minibatch,
convergence of the log-likelihoods required 50,000 iterations
(which is 2,500 epochs)
with the second-order polynomials of the years in the covariates
but only 10,000 iterations (500 epochs) without,
all at a learning rate $2 \times 10^{-2}$.
The negative log-likelihoods of the trained models for the four sets
of covariates converge to (0) 1.124, (1) 1.077, (2) 1.107, and (3) 1.061,
averaged over all samples
(there are 280 samples in total --- 14 corps for each of 20 years).
Twice the negative log-likelihood is sometimes called the ``deviance,''
which is generally defined only up to an additive constant.
The decrease in deviance when moving from covariates (0) to (1) is 26.3,
the decrease from (0) to (2) is 9.5, and from (0) to (3) is 35.3;
these deviances refer to the totals over all 280 samples,
so are 280 times the average over the samples.
The degrees of freedom corresponding to each of these changes
is less than the decrease in deviance,
indicating some minor heterogeneity in the data.
(The baseline deviance generally has no statistical interpretation
in terms of a universal distribution such as $\chi^2$;
changes in deviance when changing covariates are what matter.)

With 14 samples per minibatch, each iteration of stochastic gradient descent
takes about $1.7 \times 10^{-2}$ seconds when training in private
for any of the four sets of covariates; the smallest set (0) takes
about $0.5 \times 10^{-2}$ seconds less per iteration than the largest set (3).
When training in public, each iteration takes about $1.7 \times 10^{-4}$
seconds for any of the four sets of covariates.
Training in private thus takes about 100 times longer
than training in public on standard machines in a cluster
for software development at Facebook.

\appendix
\section{Proof of Theorem~\ref{infoleak}}
\label{mainproof}

A simple, standard computation yields that the differential entropy, in bits,
of a random variable $Y$ distributed uniformly over $[-\gamma, \gamma]$ is
\begin{equation}
\label{standard}
H(Y) = \log_2(2\gamma).
\end{equation}
Combining~(\ref{mutual}), (\ref{standard}),
and the upper-bound on $H(X+Y)$ from the following lemma yields~(\ref{worst}),
completing the proof of Theorem~\ref{infoleak}.

\begin{lemma}
\label{basiclemma}
Suppose that $X$ and $Y$ are independent scalar random variables
and $\beta$ and $\gamma$ are positive real numbers
such that $|X| \le \beta < \gamma$ and $Y$ is distributed uniformly
over $[-\gamma, \gamma]$. Then,
\begin{equation}
\label{maxent}
H(X+Y) \le \frac{\beta}{\gamma} + \log_2(2\gamma),
\end{equation}
where $H$ denotes the differential entropy measured in bits (not nats),
with equality attained in~(\ref{maxent})
when $X$ is $\beta$ times a Rademacher variate, that is,
when $X = \beta$ with probability $1/2$
and $X = -\beta$ with probability $1/2$.
\end{lemma}

\begin{proof}
Denoting the cumulative distribution function of $X$ by $F$,
the probability density function of $X+Y$ is
\begin{equation}
\label{convolution}
g(y) = \frac{1}{2\gamma} \int_{-\gamma}^{\gamma} dF(y-x)
     = \frac{1}{2\gamma} \int_{y-\gamma}^{y+\gamma} dF(x)
     = \frac{F(y+\gamma) - F(y-\gamma)}{2\gamma}.
\end{equation}
Combining $|X| \le \beta$ and the definition
of a cumulative distribution function
yields that $F(x) = 0$ for $x < -\beta$ and $F(x) = 1$ for $x > \beta$,
so~(\ref{convolution}) becomes (as diagrammed in Figure~\ref{visualization})
\begin{equation}
g(y) = \left\{\begin{array}{ll}
0, & y < -\gamma-\beta \\
\frac{F(y+\gamma)}{2\gamma}, & -\gamma-\beta < y < -\gamma+\beta \\
\frac{1}{2\gamma}, & -\gamma+\beta < y < \gamma-\beta \\
\frac{1 - F(y-\gamma)}{2\gamma}, & \gamma-\beta < y < \gamma+\beta \\
0, & y > \gamma+\beta
\end{array}\right.
\end{equation}
and the entropy in bits of $X+Y$ is
\begin{multline}
H(X+Y) = -\int_{-\infty}^{\infty} g(y) \, \log_2(g(y)) \, dy \\
       = -\frac{1}{2\gamma} \int_{\gamma-\beta}^{\gamma+\beta} \left(
         F(\gamma-y) \, \log_2\left(\frac{F(\gamma-y)}{2\gamma}\right)
       + \left(1 - F(y-\gamma)\right) \,
         \log_2\left(\frac{1 - F(y-\gamma)}{2\gamma}\right)
       \right) \, dy \\
       - \frac{1}{2\gamma} \int_{-\gamma+\beta}^{\gamma-\beta}
         \log_2\left(\frac{1}{2\gamma}\right) \, dy \\
       = -\frac{1}{2\gamma} \int_{-\beta}^{\beta} \Bigl(
         F(-y) \, \log_2(F(-y)) + \left(1 - F(y)\right) \, \log_2(1 - F(y))
       \Bigr) \, dy + \log_2(2\gamma) \\
       = -\frac{1}{2\gamma} \int_{-\beta}^{\beta} \Bigl(
         F(y) \, \log_2(F(y)) + \left(1 - F(y)\right) \, \log_2(1 - F(y))
       \Bigr) \, dy + \log_2(2\gamma).
\label{entropy}
\end{multline}

The function $p \log_2(p) + (1-p) \log_2(1-p)$ is minimal for $0 < p < 1$
at $p = 1/2$, so the integral in the right-hand side of~(\ref{entropy})
is minimal when $F(y) = 1/2$ for $|y| < \beta$, in which case~(\ref{entropy})
becomes~(\ref{maxent}) with equality attained
when $X$ is $\beta$ times a Rademacher variate.
\end{proof}

\begin{figure}
\begin{centering}

\parbox{.99\textwidth}{\includegraphics[width=.98\textwidth]{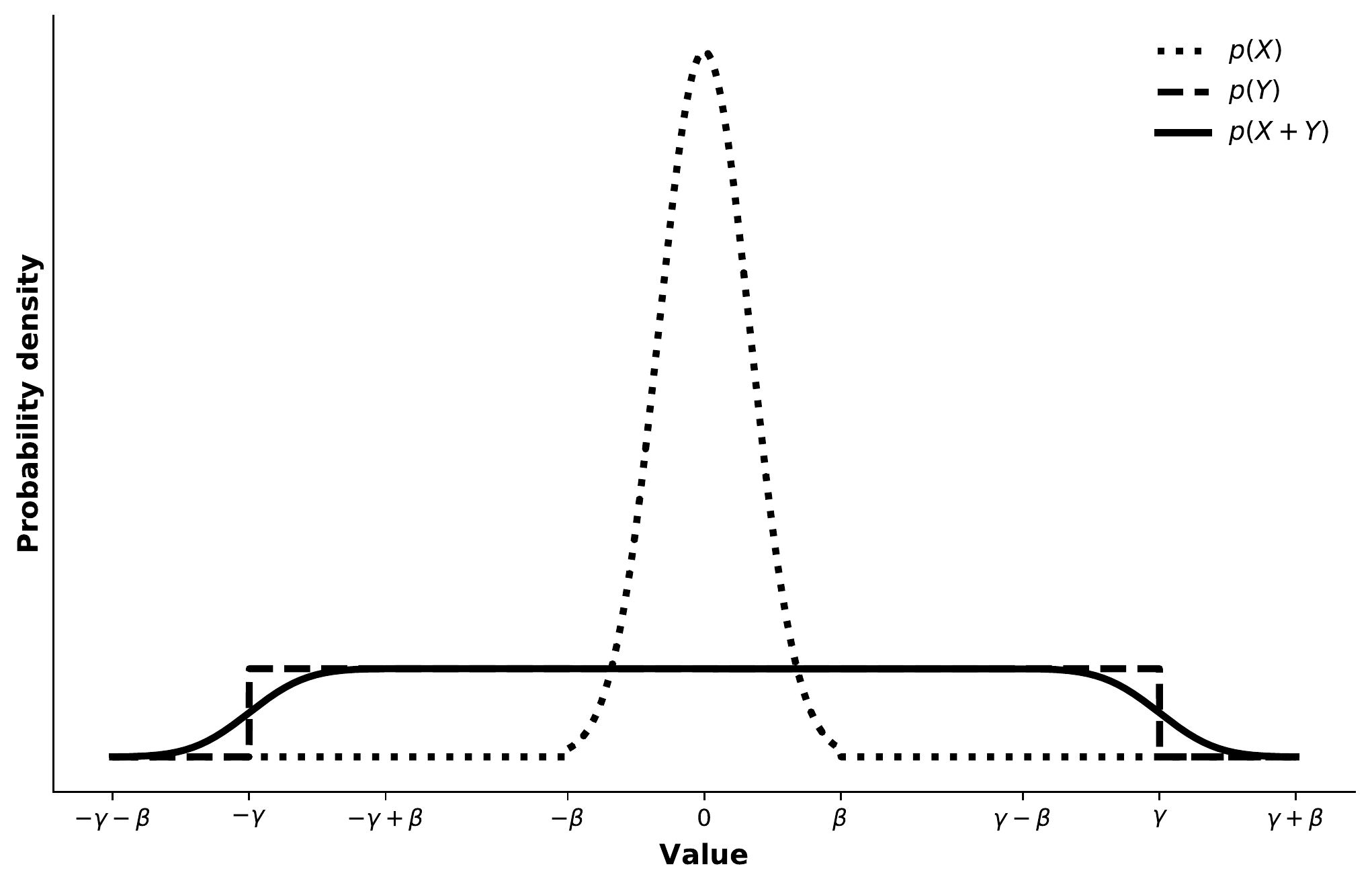}}

\end{centering}
\caption{Sketch of a visualization for understanding the proof
of Lemma~\ref{basiclemma}, in the case where the probability distribution
of $X$ is given by an even unimodal probability density}
\label{visualization}
\end{figure}

\section{Proof of Theorem~\ref{chaining}}
\label{simpleproof}

For any scalar random variables $X$, $Y$, and $Z$,
the definition of mutual information states
\begin{equation}
\label{def}
I(X;\, X+Y,\, X+Z) = H(X+Y,\, X+Z) - H(X+Y,\, X+Z \;|\; X),
\end{equation}
\begin{equation}
\label{def2}
I(X;\, X+Y) = H(X+Y) - H(X+Y \;|\; X),
\end{equation}
and
\begin{equation}
\label{def3}
I(X;\, X+Z) = H(X+Z) - H(X+Z \;|\; X),
\end{equation}
where $H$ denotes entropy.
Furthermore, the subadditivity of entropy for any arbitrary random variables
yields
\begin{equation}
\label{subentropy}
H(X+Y,\, X+Z) \le H(X+Y) + H(X+Z).
\end{equation}
Taking $X$, $Y$, and $Z$ to be independent then yields
\begin{multline}
\label{independence}
H(X+Y,\, X+Z \;|\; X) = H(Y,\, Z \;|\; X) = H(Y,\, Z) = H(Y) + H(Z) \\
= H(Y \;|\; X) + H(Z \;|\; X) = H(X+Y \;|\; X) + H(X+Z \;|\; X).
\end{multline}
Combining~(\ref{def})--(\ref{independence}) yields~(\ref{subadditivity}),
completing the proof of Theorem~\ref{chaining}.

\section{Proof of Theorem~\ref{Beaverdetails}}
\label{detailed}

We suppose that $X$, $Y$, $P$, $Q$, and $T$
are independent scalar random variables and $\gamma$ is a positive real number
such that $|X| \le 1 < 3 < \gamma$, the random variable $T$
is distributed uniformly over $[-\gamma^2, \gamma^2]$,
and $Y$, $P$, and $Q$ are distributed uniformly over $[-\gamma, \gamma]$. 
Then, since $(X-Y)-(P-Q)$ is just the difference between $X-Y$ and $P-Q$,
providing no more and no less information than $X-Y$ and $P-Q$ on their own
without $(X-Y)-(P-Q)$, and $(P^2-T) + 2(P-Q)(X-P) + (X-P)^2/2$
is also merely a deterministic combination of the observations
$P^2-T$, $P-Q$, and $X-P$, providing no more and no less information
than $P^2-T$, $P-Q$, and $X-P$ on their own
without $(P^2-T) + 2(P-Q)(X-P) + (X-P)^2/2$, the mutual information satisfies
\begin{multline}
\label{nomorenoless}
I\Bigl(X;\,
X-Y,\, P-Q,\, P^2-T,\, (X-Y)-(P-Q),\, X-P,\,
(P^2-T) + 2(P-Q)(X-P) + (X-P)^2/2\Bigr) \\
= I(X;\, X-Y,\, P-Q,\, P^2-T,\, X-P).
\end{multline}
Since $X-Q = (X-P)+(P-Q)$ and $P-Q = (X-Q)-(X-P)$
establish a bijection between the pair $X-P$ and $X-Q$
and the pair $X-P$ and $P-Q$, it follows that
\begin{equation}
\label{bijection}
I(X;\, X-Y,\, P-Q,\, P^2-T,\, X-P) = I(X;\, X-Y,\, X-Q,\, X-P,\, P^2-T).
\end{equation}
Combining~(\ref{subadditivity}) and the fact that $X$, $Y$, $Q$, $P$, and $T$
are independent yields
\begin{equation}
\label{lemmapp}
I(X;\, X-Y,\, X-Q,\, X-P,\, P^2-T)
\le I(X;\, X-Y) + I(X;\, X-Q) + I(X;\, X-P,\, P^2-T).
\end{equation}
Combining~(\ref{nomorenoless})--(\ref{lemmapp})
and (\ref{squares}) from the following lemma, together with~(\ref{worst}),
yields~(\ref{mainsquare}), completing the proof of Theorem~\ref{Beaverdetails}.
\begin{lemma}
Suppose that $X$, $P$, and $T$ are independent scalar random variables
and $\gamma$ is a positive real number such that $|X| \le 1 < 3 < \gamma$,
the random variable $P$ is distributed uniformly over $[-\gamma, \gamma]$,
and $T$ is distributed uniformly over $[-\gamma^2, \gamma^2]$. Then,
\begin{equation}
\label{squares}
I(X;\, X-P,\, P^2-T) \le I(X;\, X-P) + \frac{2}{\gamma} + \frac{1}{\gamma^2},
\end{equation}
where $I$ denotes the mutual information measured in bits.
\end{lemma}

\begin{proof}
The proof begins with a string of identities, systematically simplifying
(or re-expressing) their right-hand sides.
Indeed, since $X^2-2XP+P^2$ is simply the square of $X-P$, it follows that
\begin{equation}
\label{complication}
I(X;\, X-P,\, P^2-T) = I(X;\, X-P,\, X^2-2XP+P^2,\, P^2-T).
\end{equation}
Since $X^2-2XP+T = (X^2-2XP+P^2) - (P^2-T)$ and
$P^2-T = (X^2-2XP+P^2) - (X^2-2XP+T)$ establish a bijection
between the pair $X^2-2XP+P^2$ and $P^2-T$
and the pair $X^2-2XP+P^2$ and $X^2-2XP+T$, it follows that
\begin{equation}
I(X;\, X-P,\, X^2-2XP+P^2,\, P^2-T) = I(X;\, X-P,\, X^2-2XP+P^2,\, X^2-2XP+T).
\end{equation}
Since $X^2-2XP+P^2$ is simply the square of $X-P$, it follows that
\begin{equation}
I(X;\, X-P,\, X^2-2XP+P^2,\, X^2-2XP+T) = I(X;\, X-P,\, X^2-2XP+T).
\end{equation}
The chain rule for mutual information yields
\begin{equation}
\label{chainrule}
I(X;\, X-P,\, X^2-2XP+T) = I(X;\, X-P) + I(X;\, X^2-2XP+T \;|\; X-P).
\end{equation}

The definition of mutual information states
\begin{multline}
\label{infodef}
I(X;\, X^2-2XP+T \;|\; X-P) \\
= H(X^2-2XP+T \;|\; X-P) - H(X^2-2XP+T \;|\; X-P,\, X),
\end{multline}
where $H$ denotes the differential entropy measured in bits.
Since $T$ is independent of $X$ and $P$,
the last term in the right-hand side of~(\ref{infodef}) is
\begin{equation}
\label{conditioned}
H(X^2-2XP+T \;|\; X-P,\, X) = H(X^2-2XP+T \;|\; P, X) = H(T).
\end{equation}
The fact that $T$ is distributed uniformly over $[-\gamma^2, \gamma^2]$ yields
via a simple, straightforward calculation that
\begin{equation}
\label{exact}
H(T) = \log_2(2\gamma^2).
\end{equation}

We now upper-bound the first term in the right-hand side of~(\ref{infodef}),
by defining
\begin{equation}
S = X^2-2XP
\end{equation}
and noticing
\begin{equation}
\label{maxS}
|S| \le |X|^2 + 2|X||P| \le 1 + 2\gamma.
\end{equation}
That conditioning never increases entropy yields that the first term
in the right-hand side of~(\ref{infodef}) satisfies
\begin{equation}
\label{conditioning}
H(S+T \;|\; X-P) \le H(S+T).
\end{equation}
Combining~(\ref{maxent}) and the fact that $T$ is distributed uniformly
over $[-\gamma^2, \gamma^2]$ yields that, maximizing
over all random variables $S$ such that $|S| \le \beta = 1 + 2\gamma$
(even dropping the constraint that $S = X^2-2XP$ in the maximization),
\begin{equation}
\label{unconstrained}
H(S+T) \le \frac{1+2\gamma}{\gamma^2} + \log_2(2\gamma^2)
= \frac{2}{\gamma} + \frac{1}{\gamma^2} + \log_2(2\gamma^2).
\end{equation}

Combining~(\ref{infodef})--(\ref{unconstrained}) yields that,
maximizing over any random variable $X$ such that $|X| \le 1$,
\begin{equation}
\label{other}
I(X;\, X^2-2XP+T \;|\; X-P) \le \frac{2}{\gamma} + \frac{1}{\gamma^2}.
\end{equation}
Combining~(\ref{complication})--(\ref{chainrule}) and~(\ref{other})
yields~(\ref{squares}).
\end{proof}

\section{Chebyshev series for odd functions}
\label{chebappendix}

In this appendix, we review the approximation of odd functions
via Chebyshev series, as summarized for general (not necessarily odd) functions
in Sections~4--6 of~\cite{curry}.

Given a real-valued differentiable function $f$ on $[-z, z]$ that is odd,
that is,
\begin{equation}
f(-x) = -f(x)
\end{equation}
for any real number $x$ such that $-z \le x \le z$,
we can approximate $f$ via its Chebyshev series, as follows.
First, we select a sufficiently large positive integer $n$ and compute
\begin{equation}
\label{chebcoeffs}
c_j = \frac{2}{n} \sum_{k=1}^n
      \cos\left(\frac{j(2k-1)\pi}{4n}\right) \;
      f\left(z \cos\left(\frac{(2k-1)\pi}{4n}\right)\right)
\end{equation}
for $j = 1$, $3$, \dots, $2n-1$;
the approximation will converge as $n$ increases.
Having calculated $c_1$, $c_3$, \dots, $c_{2n-1}$ from~(\ref{chebcoeffs}),
we can compute a good approximation to $f$ evaluated at any real number $y$
such that $-z \le y \le z$:
\begin{equation}
\label{approx}
f(y) \approx \sum_{j=1}^n c_{2j-1} \, T_{2j-1}(y/z),
\end{equation}
where $T_j(x)$ denotes the Chebyshev polynomial of degree $j$
evaluated at $x = y/z$.

To calculate the right-hand side of~(\ref{approx}) efficiently
using only additions and multiplications involving the input
\begin{equation}
\label{normalized}
x = y / z,
\end{equation}
we evaluate the sums
\begin{equation}
\label{partialsums}
s_{2k-1} = \sum_{j=1}^k c_{2j-1} \, T_{2j-1}(x)
\end{equation}
for $k = 1$, $2$, \dots, $n$,
via the following recurrence:
\begin{equation}
\label{update}
t_{2k+1} = (4 x^2 - 2) t_{2k-1} - t_{2k-3}
\end{equation}
and
\begin{equation}
s_{2k+1} = s_{2k-1} + c_{2k+1} t_{2k+1},
\end{equation}
started with
\begin{equation}
t_1 = x,
\end{equation}
\begin{equation}
t_3 = (4 x^2 - 3) x,
\end{equation}
\begin{equation}
s_1 = c_1 t_1,
\end{equation}
and
\begin{equation}
s_3 = s_1 + c_3 t_3.
\end{equation}
The final sum $s_{2n-1}$ is equal to the right-hand side of~(\ref{approx}),
due to~(\ref{normalized}) and~(\ref{partialsums}).

\section{Review of stochastic gradient descent with minibatches}
\label{sgdrev}

As discussed in any standard reference on modern machine learning,
such as that of~\cite{bottou-curtis-nocedal},
minibatched stochastic gradient descent (SGD) calculates a vector $w$
of parameters that minimizes the expected value $\E(\ell(X; w))$,
where $\ell$ is a function of both $w$ and a random vector $X$.
The expected value is known as the ``risk''
and $\ell$ is known as the ``loss.''
SGD minimizes the expected loss without having direct access
to the probability distribution of $X$, instead relying solely
on samples of $X$ (usually drawn at random from a so-called ``training set'').
Minibatched SGD generates a sequence of approximations via iterations,
\begin{equation}
\label{SGD}
w^{(k+1)} = w^{(k)} - \frac{\eta}{m} \sum_{j=1}^{m}
            \frac{\partial}{\partial w} \ell(x^{(j,k)}; w)
            \biggm{|}_{w = w^{(k)}},
\end{equation}
where $\eta$ is a positive real number known as the ``learning rate,''
$\partial/\partial w$ denotes the gradient with respect to $w$
(which is $\ell$'s second argument),
$m$ is the number of samples in a so-called ``minibatch,''
and $x^{(1,k)}$, $x^{(2,k)}$, \dots, $x^{(m,k)}$ denote samples from $X$.
The iterations fail to minimize the expected loss when $\eta$ is constant
rather than decaying to 0 as the iterations proceed, but fixing $\eta$
at a sensibly small value is a common practice in machine learning
(and still ensures convergence to the minimum of the empirical risk
under suitable conditions on $\ell$,
that is, to the minimum of the average of the loss, averaged over all samples
in a fixed, finite training set).

Minimizing the regularized objective function
$\E(\ell(X; w)) + \rho \|w\|^2_2 / 2$,
where $\rho$ is a nonnegative real number
and $\|w\|_2$ denotes the Euclidean norm of $w$,
is a common way of ensuring that $w$ not become too large.
Adding such regularization to SGD is also known as ``weight decay,''
and the iterations in~(\ref{SGD}) become
\begin{equation}
w^{(k+1)} = w^{(k)} - \frac{\eta}{m} \sum_{j=1}^{m}
            \frac{\partial}{\partial w} \ell(x^{(j,k)}; w)
            \biggm{|}_{w = w^{(k)}} -\;\; \eta \rho w^{(k)};
\end{equation}
we used some weight decay for the multinomial logistic regression
of the measured data in Subsection~\ref{realdata}
(but used no weight decay for any other results reported above,
nor did we use any weight decay for iterative updates
to the so-called bias offsets in $c$ from the following appendix).

\section{Review of generalized linear models}
\label{glmrev}

As discussed in any standard reference on generalized linear models,
such as that of~\cite{mccullagh-nelder},
a generalized linear model regresses a random vector $Y$ of so-called targets
against a matrix $X$ of so-called covariates via the model
\begin{equation}
\label{glm}
g(\E(Y|X)) = Xw + c,
\end{equation}
where $g$ is known as the ``link function,''
$\E(Y|X)$ is the conditional expectation
of the vector $Y$ of targets given $X$ (the matrix of covariates),
$w$ is a vector of so-called ``weights'' (or ``parameters''),
and $c$ is a vector whose entries are independent of the values of $X$,
known as ``biases.''
Table~\ref{glmtab} lists several special cases of generalized linear models.
Given independent samples of pairs $(x^{(1)}, y^{(1)})$,
$(x^{(2)}, y^{(2)})$, \dots, $(x^{(n)}, y^{(n)})$,
the standard method of fitting the vectors $w$ of weights and $c$ of biases
is to minimize the negative of the natural logarithm of the likelihood,
that is, to minimize the empirical risk
$-\frac{1}{n} \sum_{k=1}^n \ln(p(y^{(k)} | x^{(k)}; w, c))$,
where $p(y^{(k)} | x^{(k)}; w, c)$ denotes the probability
(at the parameter values $w$ and $c$) of observing $y^{(k)}$ given $x^{(k)}$.
Minimizing $-\E(\ln(p(Y|X; w, c)))$
via the minibatched stochastic gradient descent
of the previous appendix is another (nearly equivalent) approach.

The probability distribution of $Y$ given $X$
for all of the generalized linear models considered in Table~\ref{glmtab},
except for probit regression, is a so-called ``exponential family'' of the form
\begin{equation}
\label{expfam}
p(y|x; w, c) = h(y) \exp(y^\top xw + y^\top c - \psi(xw + c)),
\end{equation}
where $h$ is a nonnegative-valued function,
and $\psi$ is known as the ``log partition function'':
$\psi(\theta) = \|\theta\|_2^2/2$ for the normal distribution,
$\psi(\theta) = \ln(1+\exp(\theta))$ for the Bernoulli distribution,
and $\psi(\theta) = \exp(\theta)$ for the Poisson distribution.
For all these cases, substituting~(\ref{expfam}) and taking the gradient
with respect to $\theta$ of both sides of $\int p(y|x; w, c) \, dy = 1$ yields
after a straightforward calculation that
$\partial \psi/\partial \theta = \E(Y|x)$, where $\theta = xw + c$;
combined with~(\ref{glm}) this yields that
$g(\partial \psi/\partial \theta) = \theta$,
so the link $g$ is the inverse of $\partial \psi/\partial \theta$.

Combining~(\ref{expfam}) and the chain rule yields that the gradient
with respect to $w$ of $\ln(p(y|x; w, c))$ is
\begin{equation}
\frac{\partial}{\partial w} \ln(p(y|x; w, c)) = x^\top \left( y
- \frac{\partial \psi}{\partial \theta}\biggm|_{\theta = xw + c} \right)
\end{equation}
and the gradient with respect to $c$ of $\ln(p(y|x; w, c))$ is
\begin{equation}
\frac{\partial}{\partial c} \ln(p(y|x; w, c)) =
y - \frac{\partial \psi}{\partial \theta}\biggm|_{\theta = xw + c}.
\end{equation}
Thus, the stochastic gradient descent of the previous appendix requires 
nothing more than addition, matrix-vector multiplications,
and evaluation of the inverse ($\partial \psi/\partial \theta$)
of the link $g$.
Needless to say, the inverse of the identity function is the identity function,
the inverse of $\ln$ is $\exp$, and the inverse of the inverse
of the cumulative distribution function $\Phi$
for the standard normal distribution is $\Phi$.
A simple calculation shows that the inverse of the logit function
$g(\mu) = \ln(\mu/(1-\mu))$ is the standard logistic function
$\partial \psi/\partial \theta = 1/(1+\exp(-\theta))$
and that the softmax detailed in Subsection~\ref{softmax} above
inverts $\ln$ applied entrywise to a probability vector
(the softmax simply applies $\exp$ entrywise and then normalizes
to form a proper probability distribution).

The probability distribution for probit regression is an exponential family,
but not of the form in~(\ref{expfam}); we handle this special case
as detailed in Sub-subsection~\ref{probitlink}.

\begin{table}
\hspace{-.65em}
\begin{tabular}{lllll}
\underline{\phantom{$g$}name\phantom{$g$}} &
\underline{\phantom{$g$}distribution\phantom{$g$}} &
\underline{\phantom{$g$}name of link\phantom{$g$}} &
\underline{\phantom{$g$}link $g(\mu)$\phantom{$g$}} &
\underline{\phantom{$g$}range of mean $\mu$\phantom{$g$}}
\vspace{.8em} \\
linear least squares & $N(\mu, I)$ & identity & $\mu$ & all real vectors
\vspace{.4em} \\
logistic regression & Bernoulli & logit & $\ln(\mu/(1-\mu))$ &
unit interval $[0, 1]$
\vspace{.4em} \\
probit regression & Bernoulli & probit &  $\Phi^{-1}(\mu)$ &
unit interval $[0, 1]$
\vspace{.4em} \\
Poisson regression & Poisson & log & $\ln(\mu)$ & nonnegative real numbers
\vspace{.4em} \\
multinomial logistic & multinomial & log performed & $\ln$ of each &
prob.\,simplex $\sum_{j=1}^k\hspace{-.1em}\mu_j\hspace{-.05em}=\hspace{-.05em}1$
\\
regression & & entry-by-entry & entry of $\mu$ &
and $\mu_1$, $\mu_2$, \dots, $\mu_k \ge 0$
\vspace{.4em}
\end{tabular}
\caption{Special cases of generalized linear models,
where $N(\mu, I)$ denotes the (possibly multi-variate) normal distribution
with mean $\mu$ and variance-covariance matrix being the identity matrix $I$,
and $\Phi$ is the cumulative distribution function
for the standard normal distribution
(so $\Phi^{-1}$ is the corresponding quantile function, the inverse of $\Phi$);
``linear least squares'' is also known as ``ordinary least squares''
or the ``general'' linear model
--- a special case of the ``generalized'' linear model}
\label{glmtab}
\end{table}

\section*{Acknowledgements}

We would like to thank Shubho Sengupta and Andrew Tulloch.

\clearpage

\bibliography{crypto}
\bibliographystyle{imaiai}

\end{document}